\documentclass{LMCS}
\usepackage{amssymb}
\usepackage{graphicx}
\usepackage{enumerate,hyperref}

\def\doi{6 (3:21) 2010}
\lmcsheading%
{\doi}
{1--27}
{}
{}
{Oct.~29, 2009}
{Sep.~\phantom07, 2010}
{}

\begin{document}

\bibliographystyle{plain}

\newcommand{\SN}{\mbox{\sf SN}}
\newcommand{\SNinf}{\mbox{\sf SN}^\omega}
\newcommand{\SNinff}{\mbox{\sf SN}^\infty}
\newcommand{\nat}{\mbox{\bf N}}
\newcommand{\real}{\mbox{\bf R}_{\geq 0}}
\newcommand{\ar}{\mbox{\sf ar}}
\newcommand{\argu}{\mbox{\sf arg}}
\newcommand{\tr}{\mbox{\sf trunc}}
\newcommand{\rt}{\mbox{\sf root}}
\newcommand{\XX}{{\bf X}}
\newcommand{\desda}{\; \Longleftrightarrow \;}
\newcommand{\gehz}{\gtrsim}
\newcommand{\nt}{\mbox{\sf not}}
\newcommand{\ms}{\mbox{\sf morse}}
\newcommand{\inv}{\mbox{\sf inv}}
\newcommand{\tl}{\mbox{\sf tail}}
\newcommand{\tlp}{\mbox{\sf tail0}}
\newcommand{\hd}{\mbox{\sf head}}
\newcommand{\Obs}{\mbox{\sf Obs}}
\newcommand{\unf}{\mbox{\sf Unf}}
\newcommand{\pob}{\mbox{\sf P}}
\newcommand{\zip}{\mbox{\sf zip}}
\newcommand{\even}{\mbox{\sf even}}
\newcommand{\odd}{\mbox{\sf odd}}
\newcommand{\alt}{\mbox{\sf alt}}
\newcommand{\fib}{\mbox{\sf Fib}}
\newcommand{\kol}{\mbox{\sf Kol}}
\newcommand{\ones}{\mbox{\sf ones}}
\newcommand{\zeros}{\mbox{\sf zeros}}
\newcommand{\ovf}{\mbox{\sf overflow}}
\newcommand{\nf}{\mbox{\bf NF}}
\newcommand{\str}{D^{\omega}}
\newcommand{\TT}{{\bf T}_s}
\newcommand{\EE}{{\bf E}}
\newtheorem{example}{Example}

\title[Well-definedness of Streams by Transformation and Termination]{Well-definedness of Streams \\ by Transformation and Termination}
\author{Hans Zantema}
\address{%
Department of Computer Science, TU Eindhoven, P.O.\ Box 513,
5600 MB Eindhoven, The Netherlands, and}
\address{\vskip-6 pt\noindent
Institute for Computing and Information Sciences, Radboud University \\
Nijmegen, P.O.\ Box 9010, 6500 GL Nijmegen, The Netherlands
}
\email{H.Zantema@tue.nl}

\begin{abstract}
Streams are infinite sequences over a given data type. A stream
specification is a set of equations intended to define a stream.

We propose a transformation from such a stream specification to a
term rewriting system (TRS) in such a way that termination of the 
resulting TRS implies that the stream specification is well-defined, that
is, admits a unique solution. As a
consequence, proving well-definedness of several interesting 
stream specifications can be done fully automatically using present
powerful tools for proving TRS termination.  

In order to increase the power of this approach, we investigate transformations
that preserve semantics and well-definedness. We give examples for which 
the above mentioned technique applies for the transformed specification 
while it fails for the original one.
 \end{abstract}

\keywords{term rewriting, stream specification}
\subjclass{F.4.2, E.1}

\maketitle\vskip1 cm

\section{Introduction}

Streams are among the simplest data types in which the objects are infinite. We
consider streams to be maps from the natural numbers to some data type $D$.
Streams have been studied extensively, e.g., in \cite{AS03}.
The basic constructor for streams is the operator `:' mapping a data element $d$
and a stream $s$ to a new stream $d:s$ by putting $d$ in front of $s$.
Using this operator we can define streams by equations. For instance, the 
stream $\zeros$ only consisting of 0's can be defined by the single equation
$\zeros = 0: \zeros$. 
More complicated streams are defined using stream functions. For instance, the
boolean {\em Fibonacci stream} $\fib$ is defined\footnote{In \cite{AS03} it is
called the infinite Fibonacci word.} as the limit of the strings
$\phi_i$ where $\phi_1 = 1, \; \phi_2 = 0, \; \phi_{i+2} =  \phi_{i+1} \phi_i$
for $i \geq 1$, showing the relationship with Fibonacci numbers. For $f$ being 
the function replacing every 0 by 1 and every 1 by 01, one easily proves by induction 
on $n$ that $f(\phi_n) = \phi_{n+1}$ for all $n \geq 1$. As $\fib$ is the limit of 
these strings, $\fib$ is a fix point of this function $f$ on boolean streams. So
the function $f$ and $\fib$ satisfy the three equations 
\[ \begin{array}{rcl}
f(0:\sigma) & = & 0:1:f(\sigma), \\ 
f(1:\sigma) & = & 0:f(\sigma), \\
\fib & = & f(\fib), \end{array}\]
for all boolean streams $\sigma$.
In this paper we consider stream specifications consisting of such a set of 
equations. We address the most fundamental question one can think of: is the intended
stream uniquely defined by these equations? More precisely, does
such a set of equations admit a unique solution as constants and
functions on streams? So in particular for $\fib$: is the boolean stream $\fib$
uniquely defined by the three equations we gave? We will call this 
{\em well-defined}, and we will show that for the equations for $\fib$ this indeed
holds. 

Although our specification of $\fib$ only consists of a few very simple equations,
the resulting stream is non-periodic and has remarkable properties. For instance, 
one can make a {\em turtle visualization} 
as follows. Choose an initial drawing direction
and traverse the elements of the stream $\fib$ as follows:
if the symbol 0 is read then the drawing direction is moved
120 degrees to the right; if the symbol 1 is read then the drawing direction is
moved 30 degrees to the left. In both cases after doing so a line of unit length
is drawn. Then after 100.000 steps the following picture is obtained.

\includegraphics[height=4.9in,width=5.4in,angle=0]{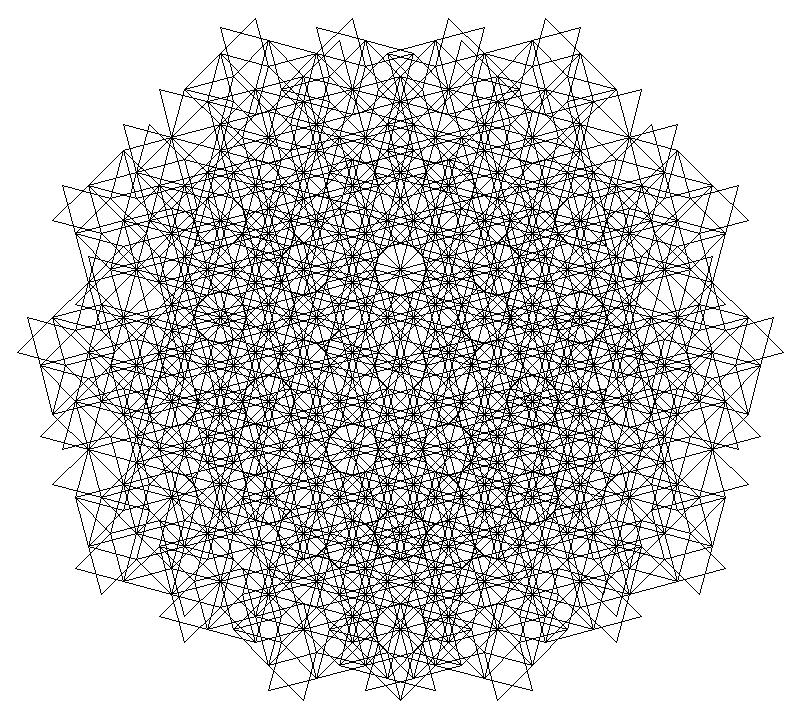}

Another turtle visualization of $\fib$ with different parameters was given in \cite{Z09}.
For turtle visualizations of similar stream specifications we refer to
\begin{center}
{\tt http://www.win.tue.nl/\~{}hzantema/str.html}.
\end{center}

To show that well-definedness does not always hold, observe that the function $f$
defined by 
\[ f(0:\sigma) = 1:f(\sigma), \;\; f(1:\sigma) = 0:f(\sigma)\]
has no fixpoints, that is, adding an equation $c = f(c)$ yields no solution for $c$.
On the other hand, the function $f$ defined by
\[ f(0:\sigma) = 0:f(\sigma), \;\; f(1:\sigma) = 1:f(\sigma)\]
is the identity with infinitely many fixpoints, yielding infinitely many solutions 
of the equation $c = f(c)$. Finally, the function $f$ defined by
\[ f(0:\sigma) = 0:1:f(\sigma), \;\; f(1:\sigma) = 1:0:f(\sigma)\]
has exactly two fixpoints: the Thue-Morse stream and its inverse.

Our approach to prove well-definedness of stream specifications is based on the 
following idea. Derive rewrite rules from the equations in such a way that by 
these rules the $n$-th element of the stream can be computed for every $n$. The 
term rewriting systems (TRS) consisting of these rules will be orthogonal by 
construction, so if the computation yields a result, this result will be unique.
So the remaining key point is to show that the computation always yields a result,
which is the case if the TRS is terminating.  The past ten years showed up a remarkable 
progress in techniques and implementations for proving termination of TRSs
\cite{AG00,GTS04,MZ07}.
One of the objectives of this paper is to exploit this power for
proving well-definedness of stream specifications. In our approach we introduce 
fresh operators $\hd$ and $\tl$ intended to observe streams. We present a 
transformation of the specification to its {\em observational variant}. This 
is a TRS mimicking the stream specification in such a way that $\hd$ or $\tl$ 
applied on any stream constant or stream function can always be rewritten. In 
particular for a stream term $t$ it serves for computing $\hd(\tl^{n-1}(t))$, 
representing the $n$-th element of $t$. So not only a proof of well-definedness is 
provided, our approach also yields an algorithm to compute the $n$-th element of any 
stream term, for any $n$.

This transformation is straightforward and easy to implement; an 
implementation for boolean stream specifications is discussed in Section
\ref{sectool}.

The main result
of this paper states that if the observational variant of a stream specification is
terminating, then the stream specification is well-defined. 
It turns out that for several
interesting cases termination of the observational variant of a specification
can be proved by termination tools like AProVE \cite{aprove} or
TTT2 \cite{ttt2}. This provides a new 
technique to prove well-definedness of stream specifications fully automatically, 
applying for cases where earlier approaches fail. Our main result appears in two 
variants: 
\begin{enumerate}[$\bullet$]
\item a variant restricting to 
ground terms for general stream specifications (Theorem \ref{main}), and 
\item a variant generalizing to all streams for stream specifications not
depending on particular data elements (Theorem \ref{main2}).
\end{enumerate}
By an example we show that the approach does not work for general stream
specifications and functions applied on all streams. Moreover, we show that 
our technique is not complete: the fix point definition of the Fibonacci stream
$\fib$ as we just gave is a well-defined stream 
specification for which the observational variant is non-terminating. However,
we will also investigate transformations from stream specifications to stream
specifications preserving 
semantics, so also preserving well-definedness. Applying such a
transformation to our 
specification of $\fib$ gives an alternative specification specifying the same
stream $\fib$, but for which the observational variant is terminating, to be 
proved automatically by a termination tool. In this way we prove 
well-definedness of $\fib$ with respect to the original stream specification.
More general, applying such semantics preserving transformations increases
the power of our approach to prove well-definedness of stream specifications.
 
Proving well-definedness in stream specification is closely related to proving equality
of streams. A standard approach for this is co-induction \cite{R05}: two streams or stream
functions are equal if a bisimulation can be found between them. Finding such an
arbitrary bisimulation is a hard problem in the general setting, but
restricting to circular co-induction \cite{GLR00,LR09} finding this automatically is tractable.
A strong tool doing so is Circ \cite{LR07,LGCR09}. The tool Circ focuses on proving equality, but
proving well-definedness of a function $f$ can also be proved by equality as long as the 
equations for $f$ are orthogonal: take
a copy $f'$ of $f$ with the same equations, and prove $f' = f$. Here orthogonality is
essential: if for instance a stream $c$ has two rules $c = 0:c$ and $c = 1:c$, then the
system is non-orthogonal and admits every boolean stream as a solution, while by having
a copy $c'$ with the same rules one can prove $c = c'$ by only using the rules $c=0:c$
and $c' = 0:c'$.

The input format of Circ differs from what we call stream specifications: in order to 
fit in the co-induction approach $\hd$ and $\tl$ are already building blocks and the 
Circ input is essentially the same as what we call the observational variant.
Our implementation
as discussed in Section \ref{sectool} offers the facility to transform
a stream specification to Circ format, and also generate the equalities representing
well-definedness in Circ format. For very
simple examples the equalities can be proved automatically by Circ, but for several
small stream specifications Circ fails while our approach succeeds in proving 
well-definedness. Conversely
our approach can be helpful to prove equality of two streams: if one stream satisfies 
the specification of the other one, and both specifications are well-defined, then
the streams are equal. 

Another closely related topic is {\em productivity} of stream specifications, as studied
by \cite{EGH08}. Productive stream specifications are always well-defined.
Conversely we will give an example (Example \ref{exwdnp})
of a stream specification that is
well-defined, but not productive. Our format of stream specifications is
strongly inspired by \cite{EGH08}. In \cite{EGH08} a technique has been developed for
establishing productivity of single ground terms fully automatically for a 
restricted class of stream
specifications. In particular, only a mild type of nesting in the right-hand
sides of the equation is allowed. If these restrictions hold, then the approach yields 
a full decision procedure for productivity, and provides a corresponding implementation
by which for a wide range of examples productivity can be
proved fully automatically. Productivity of a single ground term implies
well-definedness of that single term. On the other hand, our technique often 
applies where their
restrictions do not hold, or for proving well-definedness for systems that are not
productive. Apart from the technique from \cite{EGH08} there are more results on
productivity. An approach to prove productivity by means of
outermost termination has been presented in \cite{ZR09}; a more recent approach using
transformations and context-sensitive termination is presented in \cite{ZR10}. For both
these approaches the power of present termination provers is exploited for proving 
productivity automatically, similar to what we do in this paper for proving
well-definedness.

In \cite{H08} well-definedness of a stream specification is claimed if some particular
syntactic conditions hold, like all right-hand sides of the equations have ":" as its root.
Their result both follows from our main theorem and from the productivity analysis of
\cite{EGH08}.

Both stream equality \cite{R06} and productivity \cite{S09} have
been proved to be $\Pi^0_2$-complete, hence undecidable. By
a similar Turing machine construction the same is expected to hold
for stream well-definedness.

This paper is an extension of the RTA conference paper \cite{Z09} and the corresponding
tool description \cite{Z09a}. Compared to these papers 
\begin{enumerate}[$\bullet$]
\item some definitions have
been slightly modified in order to cover a more general setting,
\item
the process of unfolding and other transformations preserving the semantics have been 
worked out in detail in Section \ref{unf} and Section \ref{secex},
while in \cite{Z09} only some of the ideas were sketched by examples,
\item more examples are given, in particular specifying the paper folding stream 
and the Kolakoski stream.
\end{enumerate} 
The paper is structured as follows. In Section \ref{secbd} we present the basics of 
stream specifications and their models. In Section \ref{unf} we show how a non-proper
stream specification can be unfolded to a proper stream specification preserving
semantics and well-definedness.  In Section \ref{secov} we define the transformation
of a proper stream specification to its observational variant. In Section \ref{secmt} we 
present and prove the main theorem: if the observational variant is terminating then
the specification is well-defined, that is, restricted to ground terms it has a unique 
model.  In Section \ref{sectool} we describe our implementation.
In Section \ref{secdisf} we show that the restriction to ground terms in the main 
theorem may be removed in case the stream specification is data independent, that is, 
left-hand sides of equations do not contain data values.
In Section \ref{secex} we present requirements on transformations on stream 
specifications for
preserving semantics and well-definedness. In case the observational variant of a
stream specification is not terminating, or the tools fail to prove termination,
then we can apply such transformations. Often then the observational variant is
terminating, proving not only well-definedness of the transformed specification, but
also of the original one. One of the corresponding examples serves for proving 
incompleteness of our main theorem. We conclude in Section \ref{secconcl}.

\section{Streams: Specifications and Models}
\label{secbd}

In stream specifications we have two sorts: $s$ (stream) and $d$
(data). We assume the set $D$ of data elements to consist of the unique normal forms of 
ground terms over some signature $\Sigma_d$ with respect to some
terminating orthogonal TRS $R_d$ over $\Sigma_d$. Here all symbols of $\Sigma_d$ 
are of type $d^n \to d$ for some $n \geq 0$. 
We assume a particular symbol $:$ having 
type $d \times s \to s$. For giving the actual stream specification
we need a set $\Sigma_s$ of stream symbols, each being of type
$ d^n \times s^m \to s$ for $n,m \geq 0$. Now terms of sort $s$
are defined inductively as follows:
\begin{enumerate}[$\bullet$]
\item a variable of sort $s$ is a term of sort $s$,
\item if $f \in \Sigma_s$ is of type $ d^n \times s^m \to s$,
$u_1,\ldots,u_n$ are terms over $\Sigma_d$ and 
$t_1,\ldots,t_m$ are terms of sort $s$, then
$f(u_1,\ldots,u_n,t_1, \dots, t_m)$ is a term of sort $s$,
\item if $u$ is a term over $\Sigma_d$ and $t$ is a term of sort
$s$, then $u:t$ is a term of sort $s$.
\end{enumerate}
Note that we do not allow function symbols 
with output sort $d$ and input containing sort $s$. One reason for this is that we 
do not want that distinct data elements are made equal by stream equations.

An equation of sort $s$ is a pair $(\ell,r)$ of terms of sort $s$, usually 
written as $\ell = r$. An equation can also be considered as a rule in a TRS.
For basic properties of TRSs we refer to \cite{BN98}. In particular, an orthogonal 
TRS is always confluent, from which it can be concluded that every term has 
at most one normal form. Here orthogonal means that the left-hand sides of the rules
are non-overlapping, and every variable occurs at most once in any left-hand side.

As a notational convention variables of sort $d$ will be denoted by
$x,y$, terms of sort $d$ by $u,u_i$, variables of sort $s$ by $\sigma,\tau$, and
terms of sort $s$ by $t,t_i$.

\begin{defi}
\label{defss}
A {\em stream
specification} $(\Sigma_d,\Sigma_s,R_d,R_s)$ consists of  
$\Sigma_d,\Sigma_s,R_d$ as given before, and a set $R_s$ of equations
over $\Sigma_d \cup \Sigma_s \cup \{ : \}$ of sort $s$.

A stream specification $(\Sigma_d,\Sigma_s,R_d,R_s)$ is called {\em proper}
if all equations in $R_s$ are of the shape
\[ f(u_1,\ldots,u_n,t_1, \dots, t_m) = t, \]
where 
\begin{enumerate}[$\bullet$]
\item $f \in \Sigma_s$ is of type $ d^n \times s^m \to s$,
%\item for every $i = 1,\ldots,n$ the term $u_i$ is either a variable 
%of sort $d$ or $u_i \in D$, 
\item for every $i = 1,\ldots,m$ the term $t_i$ is either a variable 
of sort $s$, or $t_i = x : \sigma$ where $x$ is a variable of sort
$d$ and $\sigma$ is a variable of sort $s$,
\item $t$ is any term of sort $s$,
\item $R_s \cup R_d$ is orthogonal,
\item Every term of the shape
$f(u_1,\ldots,u_n,u_{n+1}:t_1,\ldots,u_{n+m}:t_m)$
for $f \in \Sigma_s$ of type $ d^n \times s^m \to s$, and
$u_1,\ldots,u_{n+m} \in D$ matches with the left-hand side of an equation from
$R_s$.
\end{enumerate}
\end{defi}\medskip

\noindent Some parts in this definition allow modification, but for being a basis for the rest of
our theory we fix this choice. All of our examples are on boolean streams, but by allowing
data to be ground normal forms of a data TRS, the setting is much more general.

Sometimes we call $R_s$ a stream specification: in that case $\Sigma_d$, 
$\Sigma_s$ consist of the symbols of sort $d$, $s$, respectively, occurring in 
$R_s$, and $R_d = \emptyset$. 

\begin{example}
\label{tm}
For specifying the Thue-Morse sequence the data elements are $0,1$,
and a data operation $\nt$ is used. The data rewrite system $R_d$
consists of the two rules $\nt(0) \to 1$ and $\nt(1) \to 0$. The
set $R_s$ consists of the equations
\[ \begin{array}{rclrcl}
\ms & = & 0:\zip(\inv(\ms),\tl(\ms)) \hspace{8mm} &
\tl(x:\sigma) & = & \sigma \\
\inv(x:\sigma) & = & \nt(x) : \inv(\sigma) &
\zip(x:\sigma,\tau) & = & x : \zip(\tau,\sigma) 
\end{array} \]
This is a proper stream specification.
\end{example}

Definition \ref{defss} is closely related to the definition of stream specification
in \cite{EGH08}. In fact there are two differences:
\begin{enumerate}[$\bullet$]
\item We want to specify streams for every ground term of sort $s$, while in
\cite{EGH08} there is a designated constant to be specified.
\item Our restriction on left-hand sides of $R_s$ in a proper stream specification
is stronger than the exhaustiveness from
\cite{EGH08}. However, by introducing fresh symbols and equations for defining these
fresh symbols, every stream specification in the format of \cite{EGH08} can be
unfolded to a proper stream specification in our format. This is worked out in
Section \ref{unf}.
\end{enumerate}\medskip

\noindent Stream specifications are intended to specify streams for the
constants in $\Sigma_s$, and stream functions for the other elements
of $\Sigma_s$. The combination of these streams and stream functions
is what we will call a {\em stream model}.

More precisely, a {\em stream} over $D$ is a map from the natural 
numbers to $D$. Write $\str$ for the set of all streams over $D$.
In case of $D = \emptyset$ we have $\str = \emptyset$; in case 
of $\# D = 1$ we have $\# \str = 1$. So in non-degenerate cases we
have $\# D \geq 2$.

It seems natural to require that stream functions in a stream model
are defined on all streams. However, it turns out that several desired properties
do not hold when requiring this. Therefore we allow stream functions to be
defined on some set $S \subseteq \str$ for which every ground term can be 
interpreted in $S$.

\begin{defi}
A {\em stream model} is defined to consist of a set $S \subseteq
\str$ and a set of functions $[f]$ for every $f \in \Sigma_s$, where
$[f] : D^n \times S^m \to S$ if the type of $f \in \Sigma_s$ is $d^n \times s^m \to s$.
\end{defi}

For a ground term $u$ over $\Sigma_d$ write $\nf(u)$ for its $R_d$-normal form.
We write $\TT$ for the set of ground terms of sort $s$ over $\Sigma_d \cup \Sigma_s
\cup \{:\}$. 
For $t \in \TT$ the stream 
interpretation $[t]$ in the stream model $(S,([f])_{f \in \Sigma_s})$ is 
defined inductively by:
\[ \begin{array}{rcll}
{} [f(u_1,\ldots,u_n,t_1,\ldots,t_m)] & = &  
[f]([u_1],\ldots,[u_n],[t_1],\ldots,[t_m]) & \mbox{ for $f \in \Sigma_s$} \\
{} [f(u_1,\ldots,u_n)] & = & \nf(f(u_1,\ldots,u_n))
 & \mbox{ for $f \in \Sigma_d$} \\
{} [u:t](0) & = & [u] \\
{} [u:t](i) & = & [t](i-1) & \mbox{ for $i > 0$} \end{array} \]
for all ground terms $u,u_i$ of sort $d$ and all ground terms 
$t,t_i$ of sort $s$. 

So in a stream model: 
\begin{enumerate}[$\bullet$]
\item every data operator is interpreted by its corresponding term
constructor, after which the result is reduced to normal form, 
\item every stream operator $f$ is
interpreted by the given function $[f]$, and 
\item the operator $:$ applied 
on a data element $d$ and a stream $s$ is
interpreted by putting $d$ on the first position and
shifting every stream element of $s$ to its next position. 
\end{enumerate}
Any stream model  $(S,([f])_{f \in \Sigma_s})$ can be restricted to 
a stream model  $(S',([f])_{f \in \Sigma_s})$ satisfying $S' \subseteq S$ and
$S' = \{[t] \mid t \in \TT \}$, note that from
$S' = \{[t] \mid t \in \TT \}$ we conclude that $S'$ is closed under $[f]$ for every
$f \in \Sigma_s$.

\begin{defi}
A stream model $(S,([f])_{f \in \Sigma_s})$  is said to {\em
satisfy} a stream specification $(\Sigma_d,\Sigma_s,R_d,R_s)$ if 
$[\ell \rho] = [r \rho]$ for every equation $\ell = r$
in $R_s$ and every ground substitution $\rho$. We also say that
the specification {\em admits} the model.
\end{defi}

If a stream model $(S,([f])_{f \in \Sigma_s})$ satisfies a stream specification, then
the stream model $(S',([f])_{f \in \Sigma_s})$ defined by $S' = \{[t] \mid t \in \TT \}$
satisfies the same stream specification by definition.

\begin{defi}
A stream specification is {\em well-defined} if there is exactly one 
stream model $(S,([f])_{f \in \Sigma_s})$ satisfying the stream 
specification for which $S = \{[t] \mid t \in \TT \}$.
\end{defi}

One can wonder why to restrict to $S = \{[t] \mid t \in \TT \}$. Another option 
would be simply state $S = \str$. However, sometimes restricting to ground terms 
yields a unique model, while functions applied on arbitrary streams are not unique. 
In Example \ref{exgr} we will see an example of this phenomenon.
By restricting to interpretations of ground terms and ignoring unreachable streams, 
we arrived at our definition of well-definedness. 

Not every proper stream specification is well-defined: if $\# D > 1$ and 
$R_s$ only consists of the equation $c = c$ then every stream $[c]$
satisfies the specification. Less trivial is the 
boolean stream specification
\[ c = 0:f(c), \;\;\;\;\; f(x:\sigma) = \sigma,  \]
in which $[f]$ can be chosen to be the tail function and 
$[c]$ be any stream starting with $0$, yielding several stream models.
There are also proper stream specifications with no model, for instance
\[ \begin{array}{rclrcl}
c & = & f(c), &  f(x:\sigma) & = & g(x,\sigma), \\
g(0,\sigma) & = & 1 : \sigma, &  g(1,\sigma) & = & 0 : \sigma 
\end{array} \]
Here $[f(c)]$ starts with 1 if [c] starts with 0, and conversely, contradicting
$[c] = [f(c)]$.

\section{Unfolding Stream Specifications}
\label{unf}

The specification of the function $f$
\[ f(0:\sigma) = 0:1:f(\sigma), \;\;\; f(1:\sigma) = 0:f(\sigma) \]
in the introduction to define $\fib$ does not meet the requirements of a 
proper stream specification
since the argument $0:\sigma$ in the left-hand side $f(0:\sigma)$ is not of the
right shape. Introducing a fresh symbol $g$ and unfolding yields
\[ \begin{array}{rclrcl}
f(x:\sigma) & = & g(x,\sigma) \hspace{1cm} &
g(0,\sigma) & = & 0:1:f(\sigma) \\
&&& g(1,\sigma) & = & 0:f(\sigma) \end{array} \]
satisfying the requirements of a proper stream specification. In this section
we precisely define this unfolding and show that it does not influence well-definedness.

Let $(\Sigma_d,\Sigma_s,R_d,R_s)$ be a stream specification in which $R_s$ contains an 
equation
\[ f(u_1,\ldots,u_n,t_1, \dots, t_m) = t, \]
where $f \in \Sigma_s$ is of type $ d^n \times s^m \to s$, and for some 
$i \in \{1,\ldots,m\}$ the term $t_i$ is of the shape
$t_i = u : t'$ where not both $u$ and $t'$ are variables, so the stream specification
is not proper. Then the {\em unfolded} stream specification on $f$ on position $i$, 
denoted as $\unf_{f,i}(\Sigma_d,\Sigma_s,R_d,R_s)$, is obtained by adding a fresh
symbol $g$ of type $ d^{n+1} \times s^m \to s$ to $\Sigma_s$, adding an equation
\[ f(x_1,\ldots,x_n, \sigma_1, \ldots, x_{n+1}:\sigma_i, \ldots, \sigma_m)
= g(x_1,\ldots,x_{n+1}, \sigma_1, \ldots, \sigma_m) \]
to $R_s$, where $x_{n+1}:\sigma_i$ is in the $i$-th stream position of $f$,
and in which every equation in $R_s$ of the shape
\[ f(u_1,\ldots,u_n,t_1, \ldots,u : t', \ldots t_m) = t \]
where $u : t'$ is on the $i$-th stream position of $f$, is replaced by
\[ g(u_1,\ldots,u_n,u,t_1, \ldots,t', \ldots t_m) = t, \]
where $t'$ is on the $i$-th stream position of $g$.

Applying $\unf_{f,1}$ on the $\fib$ stream specification from the introduction
yields 
\[ \begin{array}{rclrcl}
f(x:\sigma) & = & g(x,\sigma) \hspace{1cm} &
g(0,\sigma) & = & 0:1:f(\sigma) \\
\fib & = & f(\fib) & g(1,\sigma) & = & 0:f(\sigma) \end{array} \]
which is indeed a proper stream specification.

In general, for every exhaustive stream specification in the sense of \cite{EGH08}, by
repeatedly applying $\unf_{f,i}$ for various $f$, $i$, as long as an equation of the shape 
$f(u_1,\ldots,u_n,t_1, \ldots,u : t', \ldots t_m) = t$ exists for which not both
$u$ and $t'$ are variables, a proper stream specification in our sense can be obtained.

In order to justify this unfolding it remains to prove that the original stream specification
is well-defined if and only if the unfolded variant is well-defined, and in case of
well-definedness they define the same. More precisely, we prove that the
transformation $\unf_{f,i}$ {\em preserves} semantics, defined as follows. 

\begin{defi}
A transformation $\Phi$ mapping a stream specification $(\Sigma_d,\Sigma_s,R_d,R_s)$
to $(\Sigma_d,\Sigma'_s,R_d,R'_s)$ satisfying $\Sigma_S \subseteq \Sigma'_s$ is said 
to {\em preserve semantics} if 
\begin{enumerate}[$\bullet$]
\item $(\Sigma_d,\Sigma_s,R_d,R_s)$ is well-defined if and only if 
$(\Sigma_d,\Sigma'_s,R_d,R'_s)$ is well-defined, and
\item If $(\Sigma_d,\Sigma_s,R_d,R_s)$ is well-defined with corresponding model 
$(S,[\cdot])$, and $(S',[\cdot]')$ is the model corresponding to 
$(\Sigma_d,\Sigma'_s,R_d,R'_s)$, then $[t] = [t]'$ for
all ground terms of sort $s$ over $\Sigma_d \cup \Sigma_s$.
\end{enumerate}
\end{defi}\medskip

\noindent Obviously, preservation of semantics is closed under composition of such transformations.

We prove that $\unf_{f,i}$ preserves semantics in two steps:
first we only add the equation for the fresh symbol $g$, and then we do the replacement of the
$f$-equations by $g$-equations. For each of these two steps we show by a more general lemma
that semantics is preserved.

\begin{lem}
\label{lemunf1}
Let $(\Sigma_d,\Sigma_s,R_d,R_s)$ be a stream specification. Let 
$g \not\in \Sigma_s$ be of type $ d^{n+1} \times s^m \to s$.
Let $R'_s$ be the union of $R_s$ and an equation
\[ t = g(x_1,\ldots,x_{n+1}, \sigma_1, \ldots, \sigma_m) \]
in which the symbol $g$ does not occur in $t$, and $t$ does not contain variables
other than $x_1,\ldots,x_{n+1}, \sigma_1, \ldots, \sigma_m$.
Then transforming $(\Sigma_d,\Sigma_s,R_d,R_s)$ to
$(\Sigma_d,\Sigma_s \cup \{g\},R_d,R'_s)$ preserves semantics.
\end{lem}
\begin{proof}
First assume that the stream model $(S,([f])_{f \in \Sigma_s})$ satisfies the
stream specification $(\Sigma_d,\Sigma_s,R_d,R_s)$ and $S = \{[t] \mid t \in \TT
\}$. For $s_1,\ldots,s_m \in S$ choose $t_1,\ldots,t_m \in \TT$ such that 
$s_i = [t_i]$ for $i=1,\ldots,m$. Now for $d_1,\ldots,d_{n+1} \in D$ define
\[ [g](d_1,\ldots,d_{n+1},s_1,\ldots,s_m) = [t \rho] \]
for $\rho$ defined by $\rho(x_i) = d_i$ for $i=1,\ldots,n+1$ and 
$\rho(\sigma_i) = s_i$ for $i=1,\ldots,m$. Due to the compositional shape of the 
definition of $[f]$ for $f \in \Sigma_s$ this definition of $g$ is independent of
the choice of $t_1,\ldots,t_m \in \TT$. By construction this yields a stream
model $(S,([f])_{f \in \Sigma_s \cup \{g\}})$ satisfying 
$(\Sigma_d,\Sigma_s \cup \{g\},R_d,R'_s)$ and $S = \{[t] \mid t \in \TT \}$,
where in the latter $\TT$ stands for ground terms including the symbol $g$.

Conversely, assume we have a stream
model $(S,([f])_{f \in \Sigma_s \cup \{g\}})$ satisfying 
$(\Sigma_d,\Sigma_s \cup \{g\},R_d,R'_s)$ and $S = \{[t] \mid t \in \TT \}$.
Then by ignoring $g$ it is also a stream model satisfying $(\Sigma_d,\Sigma_s,R_d,R_s)$.
Due to the shape of the equation containing $g$, for every ground term $t'$ 
containing $g$ there is a ground term $t''$ not containing $g$ satisfying $[t''] = [t']$.
So we also have $S = \{[t] \mid t \in \TT \}$ for $\TT$ standing for the ground 
terms not containing $g$.

Summarizing, a model for $(\Sigma_d,\Sigma_s,R_d,R_s)$ yields a model for
$(\Sigma_d,\Sigma_s \cup \{g\},R_d,R'_s)$ and conversely, keeping the same set $S$
and both satisfying $S = \{[t] \mid t \in \TT \}$. This proves the first requirement
of semantics preservation. The second requirement holds since the
interpretations of ground terms are the same in both models.
\end{proof}

In Lemma \ref{lemunf1} the signature was extended by a fresh symbol, while except
for adding one equation for this fresh symbol, the equations remained the same.
In the next lemma it is the other way around: now the signature remains the same
and the equations may be modified. For a set $R$ of equations we write $=_R$ for 
the congruence generated by $R$, that is, the closure of $R$ under substitutions,
contexts, reflexivity, symmetry and transitivity.

\begin{lem}
\label{lemunf2}
Let $(\Sigma_d,\Sigma_s,R_d,R_s)$ and $(\Sigma_d,\Sigma_s,R_d,R'_s)$ be stream 
specifications satisfying $\ell =_{R'_s} r$ for all $\ell = r$ in $R_s$, and
$\ell =_{R_s} r$ for all $\ell = r$ in $R'_s$. 
Then transforming $(\Sigma_d,\Sigma_s,R_d,R_s)$ to
$(\Sigma_d,\Sigma_s,R_d,R'_s)$ preserves semantics.
\end{lem}
\begin{proof}
From the given connection between $R_s$ and $R'_s$ it is immediate that a model 
satisfies $(\Sigma_d,\Sigma_s,R_d,R_s)$ if and only if it 
satisfies $(\Sigma_d,\Sigma_s,R_d,R'_s)$. From this the lemma follows.
\end{proof}
 
\begin{thm}
\label{thmunf}
The transformation $\unf_{f,i}$ preserves semantics on stream specifications
on which it is defined.
\end{thm}
\begin{proof}
The operation $\unf_{f,i}$ consists of two steps: the addition of an equation
generating $g$ and the replacement of existing equations for $f$. The addition
preserves semantics due to Lemma \ref{lemunf1}. For the replacement
Lemma \ref{lemunf2} applies, for both directions applying the equation
\[ f(x_1,\ldots,x_n, \sigma_1, \ldots, x_{n+1}:\sigma_i, \ldots, \sigma_m)
= g(x_1,\ldots,x_{n+1}, \sigma_1, \ldots, \sigma_m). \]
As both transformations preserve semantics, the same holds for the composition
$\unf_{f,i}$. 
\end{proof}

\section{The Observational Variant}
\label{secov}

We define a transformation $\Obs$ transforming the 
original set of equations $R_s$ in a proper stream specification to 
its {\em observational variant} $\Obs(R_s)$, being a TRS. The basic idea is that the streams are observed by
two auxiliary operators $\hd$ and $\tl$, of which $\hd$ picks the
first element of the stream and $\tl$ removes the first element
from the stream, and that for every $t \in \TT$  of type stream both 
$\hd(t)$ and $\tl(t)$ can be rewritten by $\Obs(R_s)$.

The main result of this paper is that if $\Obs(R_s) \cup R_d$ is 
terminating for a given proper stream specification $(\Sigma_d,\Sigma_s,R_d,R_s)$,
then the specification is well-defined, that is, it admits a unique model 
$(S,([f])_{f \in \Sigma_s})$
satisfying $S = \{[t] \mid t \in \TT \}$. As a consequence, the
specification uniquely defines a corresponding stream $[t]$ for
every $t \in \TT$.

We define $\Obs(R_s)$ in two steps. First we define $\pob(R_s)$
obtained from $R_s$ by modifying the equations as follows. By
definition every equation of $R_s$ is of the shape
\[ f(u_1,\ldots,u_n,t_1, \dots, t_m) = t \]
where for every $i = 1,\ldots,m$ the term $t_i$ is either a variable 
of sort $s$, or $t_i = x : \sigma$ where $x$ is a variable of sort
$d$ and $\sigma$ is a variable of sort $s$. In case $t_i = x :
\sigma$ then in the left-hand side of the equation the subterm $t_i$ is 
replaced by $\sigma$, while in the right-hand side of the equation
every occurrence of $x$ is replaced by $\hd(\sigma)$ and 
every occurrence of $\sigma$ is replaced by $\tl(\sigma)$.

For example, the equation for $\zip$ in Example \ref{tm} will be replaced
by
\[ \zip(\sigma,\tau)  \to  \hd(\sigma) : \zip(\tau,\tl(\sigma)). \] 

Now we are ready to define $\Obs$.

\begin{defi}
Let $(\Sigma_d,\Sigma_s,R_d,R_s)$ be a proper stream specification; $\tl \not\in \Sigma$. 
Let $\pob(R_s)$ be defined as above. Then $\Obs(R_s)$ is the TRS over
$(\Sigma_d \cup \Sigma_s) \cup \{:, \hd, \tl \}$ consisting of
\begin{enumerate}[$\bullet$]
\item the two rules 
\[ \hd(x: \sigma) \to x, \;\;\; \tl(x:\sigma) \to \sigma,  \]
\item for every rule in $\pob(R_s)$ of the shape $\ell \to u : t$
the two rules 
\[ \hd(\ell) \to u, \;\;\; \tl(\ell)  \to t,  \]
\item for every rule in $\pob(R_s)$ of the shape $\ell \to r$ with
$\rt(r) \neq \; : \;$ the two rules 
\[ \hd(\ell) \to \hd(r), \;\;\; \tl(\ell) \to \tl(r). \]
\end{enumerate}
\end{defi}\medskip

\noindent The reason for first transforming $R_s$ to $\pob(R_s)$ is
that for the validity of the main theorem we need the special shape of
the rules of $\Obs(R_s)$ in which apart from the root symbol $\hd$ or
$\tl$ and one symbol from $\Sigma_s$, every left-hand sides only
consists of variables.

\begin{example}
\label{tmo}
For the set $R_s$ of equations given in Example \ref{tm} we rename the symbol
$\tl$ by $\tlp$ in order to keep the symbol $\tl$ for the fresh
symbol introduced in the $\Obs$ construction. Then the TRS $\Obs(R_s)$
consists of the following rules:
\[ \begin{array}{rclrcl}
\hd(x: \sigma) & \to & x &
\hd(\tlp(\sigma)) & \to & \hd(\tl(\sigma)) \\
\tl(x:\sigma) & \to & \sigma &
\tl(\tlp(\sigma)) & \to & \tl(\tl(\sigma)) \\
\hd(\ms) & \to & 0 &
\hd(\zip(\sigma,\tau)) & \to & \hd(\sigma) \\
\tl(\ms) & \to & \zip(\inv(\ms),\tlp(\ms)) \hspace{5mm} &
\tl(\zip(\sigma,\tau)) & \to & \zip(\tau,\tl(\sigma)) \\
\hd(\inv(\sigma)) & \to & \nt(\hd(\sigma)) \\
\tl(\inv(\sigma)) & \to & \inv(\tl(\sigma)) \\
\end{array} \]
Together with the rules $\nt(0) \to 1$ and $\nt(1) \to 0$ from
$R_d$ this TRS is terminating as can easily be proved fully
automatically by AProVE \cite{aprove} or TTT2 \cite{ttt2}. 
As a consequence, the result of
this paper states that the specification uniquely defines a stream
for every ground term of type $s$, in particular for $\ms$. 
\end{example}

\section{The Main Theorem}
\label{secmt}

We start this section by presenting our main theorem.

\begin{thm}
\label{main}
Let $(\Sigma_d,\Sigma_s,R_d,R_s)$ be a proper stream specification for
which the TRS $\Obs(R_s) \cup R_d$ is terminating.  Then the stream
specification is well-defined.
\end{thm}

Recall that a stream specification is defined to be well-defined if it
admits a unique model $(S,([f])_{f \in \Sigma_s})$
satisfying $S = \{[t] \mid t \in \TT \}$. 
Before proving the theorem we show by an example why it is
essential to restrict to $S = \{[t] \mid t \in \TT \}$ rather than
choosing $S = \str$. A degenerate example is obtained if there are no
constants of sort $s$, and hence $\TT = \emptyset$. More
interesting is the following.

\begin{example}
\label{exgr}
Consider the proper boolean stream specification with $R_d = \emptyset$
and $R_s$ consists of: 
\[ \begin{array}{rclrcl}
c & = & 1 : c & \hspace{1cm}
f(x:\sigma) & = & g(x,\sigma) \\
g(0,\sigma) & = & f(\sigma) \\
g(1,\sigma) & = & 1 : f(\sigma)
\end{array} \]
obtained by unfolding
\[ \begin{array}{rclrcl}
c & = & 1 : c & \hspace{1cm}
f(0:\sigma) & = & f(\sigma) \\
f(1:\sigma) & = & 1 : f(\sigma)
\end{array} \]
The function $f$ has been specified in such a way that it tries to 
remove all 0's from its argument. So for streams specified by terms like $f(c)$ 
there is nothing to remove, and we expect well-definedness: the term $f(c)$ will 
uniquely be defined to be the stream of only ones. However, for streams
containing only finitely many 1's this may be problematic. Note that
by the symbols $c$, $:$, $0$ and $1$ only the streams with
finitely many 0's can be constructed, so for ground terms over the symbols occurring in
the specification this problem
does not arise. Indeed, it turns out that the TRS $\Obs(R_s) \cup R_d$ is terminating,
so by Theorem \ref{main} the specification is well-defined.
It is interesting to remark that the approach from \cite{EGH08} fails to
prove productivity, as this stream 
specification is not \emph{data-obliviously productive}, i.e., the  
identity of the data is essential for productivity. Moreover, also 
Circ \cite{LGCR09} fails to prove well-definedness of this stream
specification.

We concluded that this example is well-defined: it admits a 
unique model $(S,([f])_{f \in \Sigma_s})$
satisfying $S = \{[t] \mid t \in \TT \}$. However, when extending
to all streams the function $[f] : \str \to \str$ is not uniquely
defined, even if we strengthen the requirement of $[\ell \rho] = [r
\rho]$ for all equations $\ell = r$ and all ground substitutions
$\rho$ to an open variant in which the $\sigma$'s in the equations are
replaced by arbitrary streams. Write $\ones$ and $\zeros$ for the
streams only consisting of ones, resp. zeros. Two distinct models 
$[\cdot]_1$ and $[\cdot]_2$ satisfying the stream specification are
defined by:
\[ [c]_1 = [f]_1(s) = [g]_1(u,s) = \ones \mbox{ for all $s \in \str,
u \in D$}, \]
and $[c]_2 = \ones$, and $[f]_2(s) = [g]_2(u,s) = \ones$
for $u \in D$ and streams $s$ containing infinitely many ones, and
$[f]_2(s) = 1^n: \zeros, \; [g]_2(u,s) = [f]_2(u:s)$
for $u \in D$ and streams $s$ containing $n < \infty$ ones.
\end{example}

Now we arrive at the proof of Theorem \ref{main}. The plan of the
proof is as follows.
\begin{enumerate}[$\bullet$]
\item First we construct a function $[\cdot]_1 : \TT \to \str$, and
choose $S_1 = \{[t]_1 \mid t \in \TT \}$.
\item Next we show that if $[t_i]_1 = [t'_i]_1$ for $i = 1,\ldots,m$,
then 
\[[f(u_1,\ldots,u_n,t_1,\ldots,t_m)]_1 = 
[f(u_1,\ldots,u_n,t'_1,\ldots,t'_m)]_1,\]
by which $[f]_1$ is
well-defined and we have a model $(S_1,([f]_1)_{f \in \Sigma_s})$.
\item We show this model satisfies the specification.
\item We show that no other model $(S,([f])_{f \in \Sigma_s})$ with 
$S = \{[t] \mid t \in \TT \}$ satisfies the specification.
\end{enumerate}
First we define $[t]_1 \in \str$ for any $t \in \TT$. Since elements
of $\str$ are functions from $\nat$ to $D$, a function $[t]_1 \in
\str$ is defined by defining $[t]_1(n)$ for every $n \in \nat$. Due to
the assumption of the theorem the TRS $\Obs(R_s) \cup R_d$ is terminating. 
According to the definition of a proper stream specification the TRS 
$R_s \cup R_d$ is orthogonal, and by the construction $\Obs$ the
TRS $\Obs(R_s) \cup R_d$ is orthogonal, too. So it is confluent. Since we 
assume termination, we conclude that every ground term of
sort $d$ has a unique normal form with respect to $\Obs(R_s) \cup
R_d$. 

Assume such a normal form of sort $d$ contains a symbol from $\Sigma_s \cup \{:\}$.
Choose such a symbol with minimal position, that is, closest to the 
root.  Since the term is of sort $d$, this symbol is not the root. 
Hence it has a parent. Due to minimality of position, this parent
is either $\hd$ or $\tl$. Due to the shape of the rules of
$\Obs(R_s)$, a rule of $\Obs(R_s)$ is applicable on this parent 
position, contradicting the normal form assumption. So the normal
form only contains symbols from $\Sigma_d$. Since it is also a
normal form with respect to $R_d$, such a normal form is an element
of $D$. Now for $t \in \TT$ and $n \in \nat$ we define
\begin{quote}
$[t]_1(n) = $ the normal form of $\hd(\tl^n(t))$ with respect to
$\Obs(R_s) \cup R_d$,  
\end{quote}
in this way defining $[t]_1 \in \str$.

\begin{lem}
\label{lemmodel}
Let $\Obs(R_s) \cup R_d$ be terminating.
Let $f \in \Sigma_s$ of type $d^n \times s^m \to s$. Let
$u_1,\ldots,u_n \in D$ and $t_1,\ldots,t_m,t'_1,\ldots,t'_m \in \TT$ 
satisfying $[t_i]_1 = [t'_i]_1$ for $i = 1,\ldots,m$. Then
\[ [f(u_1,\ldots,u_n,t_1,\ldots,t_m)]_1 = 
[f(u_1,\ldots,u_n,t'_1,\ldots,t'_m)]_1. \]
\end{lem}
\begin{proof}
First we extend the definition of $[\cdot]_1$ to all ground terms over
$\Sigma_s \cup \Sigma_d \cup \{:, \hd,\tl\}$. For ground terms $t$ of sort $s$ 
we define it by $[t]_1(n) = $ the normal form of $\hd(\tl^n(t))$ with respect to
$\Obs(R_s) \cup R_d$, and for ground terms $u$ of sort $d$ we define $[u]_1$ to 
be the normal form of $u$ with respect to $\Obs(R_s) \cup R_d$. We prove the 
following claim.
\begin{quote} 
{\bf Claim 1:} Let $[t]_1=[t']_1$ for $t,t' \in \TT$. Let $T$ be a ground term over
$\Sigma_s \cup \Sigma_d \cup \{:, \hd,\tl\}$ of sort $s$ 
containing $t$ as a subterm. Let $T'$ be obtained from $T$ by replacing zero
or more occurrences of the subterm $t$ by $t'$. Then
\[ [\hd(T)]_1 = [\hd(T')]_1. \]
\end{quote} 
Let $>$ be the well-founded order on ground terms being the strict
part of $\geq$ defined by 
\[ v \geq v' \desda v' \mbox{ is a subterm of $v''$ such that } 
v \to_{\Obs(R_s) \cup R_d}^* v''.\]
We prove the claim for every such term $\hd(T)$ by induction on $>$.

Claim 1 is trivial if $t=T$, so we may assume that 
$T = f(u_1,\ldots,u_n, t_1,\ldots,t_m)$ such that $t$ occurs in
$u_1,\ldots,u_n, t_1,\ldots,t_m$, 
and either $f \in \Sigma_s \cup \{:,\tl\}$, and
$T' = f(u'_1,\ldots,u'_n, t'_1,\ldots,t'_m)$. For every subterm of
$u_i$ of the shape $\hd(\cdots)$ we may apply the induction
hypothesis, yielding $[u_i]_1 = [u'_i]_1 = d_i$ for all $i$,
defining $d_i \in D$.

In case the root of $T$ is not $\tl$ we rewrite 
\[\hd(T) \to^*_{\Obs(R_s) \cup R_d}
\hd(f(d_1,\ldots,d_n, t_1,\ldots,t_m),\]
and then continue by the rule 
$\hd(f(\cdots)) \to \cdots$ in $\Obs(R_s)$, yielding a term $U$ of 
sort $d$. As $\hd$ is the only symbol of sort $d$ having an argument of sort $s$, the 
only way such a term can contain $t$ 
as a subterm is by $U = C[\hd(V_1),\ldots,\hd(V_k)]$ where $t$ is 
a subterm of some of the $V_i$ and $C$ is composed from $\Sigma_d$. Similarly, we obtain
\[\hd(T') \to^*_{\Obs(R_s) \cup R_d}
\hd(f(d_1,\ldots,d_n, t'_1,\ldots,t'_m) \to C[\hd(V'_1),\ldots,\hd(V'_k)],\]
for $V_i'$ obtained from $V_i$ by replacing zero ore more occurrences 
of $t$ by $t'$. By the induction hypothesis we obtain $[\hd(V_i)]_1 = [\hd(V'_i)]_1$. So
$[\hd(V_i)$ and $[\hd(V'_i)$ rewrite to the same normal form for all $i$. Hence
\[[\hd(T)]_1 = [C[\hd(V_1),\ldots,\hd(V_k)]]_1 = 
[C[\hd(V'_1),\ldots,\hd(V'_k)]]_1 = [\hd(T')]_1,\]
which we had to prove.

In case the root of $T$ is $\tl$ then write 
\[T = \tl^i(f(\cdots)) \to^*_{\Obs(R_s) \cup R_d} 
\tl^i(f(d_1,\ldots,d_n, t_1,\ldots,t_m)\]
for $f \in \Sigma_s \cup \{:\}$. This can be rewritten by the rule 
$\tl(f(\cdots)) \to \cdots$ in $\Obs(R_s)$, yielding $V$. Note that for applicability of
this rule it is essential that the arguments of $f$ in the left-hand side are variables, 
which was achieved by first applying the transformation $\pob$.

On the same position
using the same rule we can rewrite $T' \to_{\Obs(R_s)} V'$ for $V'$ obtained
from $V$  by replacing one or more occurrences of $t$ by $t'$. 
Applying the induction hypothesis gives
$[\hd(V)]_1 = [\hd(V')]_1$ yielding
\[ [\hd(T)]_1 = [\hd(V)]_1 = [\hd(V')]_1 = [\hd(T')]_1,\]
concluding the proof of Claim 1.
\begin{quote} 
{\bf Claim 2:} Let $[t]_1=[t']_1$ for $t,t' \in \TT$. Let $T$ be a ground term over
$\Sigma_s \cup \Sigma_d \cup \{:, \hd,\tl\}$ of sort $s$ 
containing $t$ as a subterm. Let $T'$ be obtained from $T$ by replacing one
or more occurrences of the subterm $t$ by $t'$. Then $[T]_1 = [T']_1$.
\end{quote} 
Claim 2 easily follows from Claim 1 and the observation
\[ [T]_1 = [T']_1 \desda \forall i \in \nat : [\hd(\tl^i(T))]_1 = [\hd(\tl^i(T'))]_1.  \]
Now the lemma follows by applying Claim 2 and replacing $t_i$ by $t'_i$ 
successively for $i=1,\ldots,m$.
%\qed
\end{proof}

Define $S_1 = \{ [t]_1 \mid t \in \TT \}$. For any 
$f \in \Sigma_s$ of type $d^n \times s^m \to s$ for
$u_1,\ldots,u_n \in D$ and $t_1,\ldots,t_m,t'_1,\ldots,t'_m \in \TT$ 
we now define $[f]_1 : D^n \times S^m \to S$ by
\[ [f]_1(u_1,\ldots,u_n,[t_1],\ldots,[t_m]) = 
 [f(u_1,\ldots,u_n,t_1,\ldots,t_m)]_1; \]
Lemma \ref{lemmodel} implies that this is well-defined: the result is 
independent of the choice of the representants in $[t_i]_1$.
So $(S_1,([f]_1)_{f \in \Sigma_s})$ is a model.

Next we will prove that it satisfies the specification, and
essentially is the only one doing so. 

\begin{lem}
\label{lem2}
Let $\ell \to r \in R_s$ and let $\rho$ be a substitution.
Then 
\begin{enumerate}[$\bullet$]
\item there is a term $t$ such that 
$\hd(\ell \rho) \to^*_{\Obs(R_s)} t$ and
$\hd(r \rho) \to^*_{\Obs(R_s)} t$, and
\item there is a term $t$ such that 
$\tl(\ell \rho) \to^*_{\Obs(R_s)} t$ and
$\tl(r \rho) \to^*_{\Obs(R_s)} t$.
\end{enumerate}
\end{lem}
\begin{proof}
Let $f$ be the root of $\ell$. Define $\rho'$ by $\sigma \rho' =
x \rho : \sigma \rho$ for every argument of the shape $x : \sigma$
of $f$ in $\ell$, and $\rho'$ coincides with $\rho$ on all other
variables. Then $\hd(\ell \rho) = \ell' \rho'$ for some rule in
$\ell' \to r'$ in $\Obs(R_s)$. Now a common reduct $t$ of 
$r' \rho'$ and $\hd(r \rho)$ is obtained by applying the rule
$\hd(x: \sigma) \to x$ zero or more times. This yields
$\hd(\ell \rho) = \ell' \rho' \to_{\Obs(R_s)} r' \rho'
\to^*_{\Obs(R_s)} t$ and $\hd(r \rho) \to^*_{\Obs(R_s)} t$. The
argument for $\tl(\ell \rho)$ and $\tl(r \rho)$ is similar.
%\qed
\end{proof}

\begin{lem}
\label{lem3}
The model $(S_1,([f]_1)_{f \in \Sigma_s})$ satisfies the
specification $(\Sigma_d,\Sigma_s,R_d,R_s)$.
\end{lem}
\begin{proof}
We have to prove that 
$[\ell \rho]_1(i) = [r \rho]_1(i)$ for every equation $\ell = r$
in $R_s$, every ground substitution $\rho$ and every $i \in \nat$. 
By definition $[\ell \rho]_1(i)$ is the unique normal form with
respect to $\Obs(R_s) \cup R_d$ of $\hd(\tl^i(\ell \rho))$, and
$[r \rho]_1(i)$ is the similar normal form of $\hd(\tl^i(r \rho))$.
The terms $\hd(\tl^i(\ell \rho))$ and $\hd(\tl^i(r \rho))$ have a common 
$\Obs(R_s)$-reduct. For $i=0$ this follows from the first part of 
Lemma \ref{lem2}, for $i>0$ this follows from the second part of Lemma \ref{lem2}.
As they have a common reduct, their unique normal forms $[\ell \rho]_1(i)$
and  $[r \rho]_1(i)$ with respect to $\Obs(R_s) \cup R_d$ are equal, which we had to
prove.
%\qed
\end{proof}

For concluding the proof of Theorem \ref{main} we have to prove
that $(S_1,([f]_1)_{f \in \Sigma_s})$ is the only model satisfying
the specification $(\Sigma_d,\Sigma_s,R_d,R_s)$ and $S = \{[t] \mid
t \in \TT \}$. This follows from the following lemma.

\begin{lem}
\label{lem4}
Let $(S,([f])_{f \in \Sigma_s})$ be any model satisfying 
$(\Sigma_d,\Sigma_s,R_d,R_s)$, and $t \in \TT$.
Then $[t] = [t]_1$.
\end{lem}
\begin{proof}
By definition in the model for $u \in D$ and $s \in S$  we have
\[ ([:](u,s))(0) = u, \;\; ([:](u,s))(i) = s(i-1) \mbox{ for $i > 0$}. \]
In the original stream specification the symbols $\hd, \tl$ do not
occur, for these fresh symbols we now define functions $[\hd]$ and $[\tl]$ 
on streams $s$ by
\[ [\hd](s) = s(0), \;\; ([\tl](s))(i) = s(i+1) \mbox{ for $i \geq 0$}. \]
If $S \neq \str$ then it is not clear whether $[\tl](s) \in S$ for 
every $s \in S$. Therefore we extend $S$ to $\str$ and define $[f](\cdots)$ to be
any arbitrary value if at least one argument is in $\str \setminus S$; note
that for the model satisfying the specification we only required $[\ell \rho] =
[r \rho]$ for ground substitutions to $\TT$ by which these junk values do not
play a role. 

Due to the definitions of $[:]$, $[\hd]$ and $[\tl]$ 
this extended model satisfies the equations
\[ \EE = \left\{ \begin{array}{rcl}
\hd(x:\sigma) & = & x \\
\tl(x:\sigma) & = & \sigma \\
\sigma & = & \hd(\sigma) : \tl(\sigma)
\end{array} \right. \]
that is, for $\rho$ mapping $x$ to any term of sort $d$ and $\sigma$ to any term
of sort $s$ we have $[\ell \rho] = [r \rho]$ for every $\ell \to r \in \EE$.
From the definition of $\Obs(R_s)$ it is easily checked that 
any innermost step $t \to_{\Obs(R_s)} t'$ on a ground term $t$ is either an 
application of one of the first two rules of $\EE$, or it is of the shape
\[ t \to_\EE^* \cdot \to_{R_s} \cdot \to_\EE^* t'\]
where due to the innermost requirement the redex of the $\to_{R_s}$ step does not
contain the symbols $\hd$ or $\tl$ so is in $\TT$. Since the model is assumed to
satisfy the specification $(\Sigma_d,\Sigma_s,R_d,R_s)$, we conclude that $[t] =
[t']$ for every innermost ground step $t \to_{\Obs(R_s)} t'$. 

For the lemma we have to prove that 
$[t](i) = [t]_1(i)$ for every $i \in \nat$. By definition $[t]_1(i)$ is the
normal form with respect to $\Obs(R_s) \cup R_d$ of $\hd(\tl^i(t))$. Now consider 
an innermost $\Obs(R_s) \cup R_d$-reduction of $\hd(\tl^i(t))$ to $[t]_1(i)$.
By the above observation and the definitions of $[\hd]$ and $[\tl]$ we conclude that 
\[ [t](i) = [\hd(\tl^i(t))] = [[t]_1(i)] = [t]_1(i), \]
the last step since $[t]_1(i) \in D$. This concludes the proof, both of the lemma
and Theorem \ref{main}.
%\qed
\end{proof}

We conclude this section by an example of a well-defined proper stream 
specification that is not productive.

\begin{example}
\label{exwdnp}
Choose $\Sigma_s = \{c,f,g\}$, $\Sigma_d = \{0,1\}$, $R_d = \emptyset$, and $R_s$
consists of the following equations:
\[ \begin{array}{rcl}
c & = & 1 : c\\
f(x:\sigma) & = & g(f(\sigma)) \\
g(x:\sigma) & = & c.
\end{array} \]
This is a valid proper stream specification for which $\Obs(R_s)$ is terminating, as
can be shown by AProVE \cite{aprove} or TTT2 \cite{ttt2}. 
Hence by Theorem \ref{main} it is well-defined.
So the ground term $f(c)$ has a unique interpretation: the stream
only consisting of 1's. However, $f(c)$ is not productive, as it only reduces to terms 
having $f$ or $g$ on top.

So the TRS $R_s$ uniquely defines $f(c)$, but is not suitable to compute its 
interpretation.
\end{example}

\section{Implementation}
\label{sectool}

In {\tt http://www.win.tue.nl/\~{}hzantema/str.zip} we offer a prototype
implementation automating proving well-definedness of boolean stream 
specifications by the approach we proposed. The main 
feature is to generate the observational variant for any given boolean stream
specification. Being only a prototype, the focus is on testing simple examples as 
they occur in this paper. The default version runs under Windows with a graphical user 
interface and provides the following features:
\begin{enumerate}[$\bullet$]
\item Boolean stream specifications can be entered, loaded, edited and stored. 
The format is the same as given here, with the only difference that for the 
operator ':' a prefix notation is chosen, in order to be consistent with the 
user defined symbols.
\item By clicking a button the observational variant of the current stream 
specification is tried to be created. In doing so, all requirements of the 
definition of
stream specification are checked. If they are not fulfilled, an appropriate error
message is shown. 
\item If all requirements hold, then the resulting observational variant is shown 
on the screen by which it can be 
entered by cut and paste in a termination tool. Alternatively, it can be stored.
\item Alternatively, the stream specification can be transformed to Circ format.
This occurs in two variants: 
\begin{enumerate}[$\bullet$]
\item a basic variant in which the Circ proof goal should be
added manually, and 
\item a version generating two copies of the specification and
generating goals for these to be equivalent.
\end{enumerate}
Again it is shown on the screen with
cut and paste facility, or the result can be stored, both for entering the result in
the tool Circ.
\item A term can be entered, and an initial part of the stream represented by this 
term can be computed. 
\item For unary symbols the process of unfolding as described in Section \ref{unf}
is supported.
\item Several stream specifications, including the Fibonacci stream (the variant
as we will present in Example \ref{fibex}), the Thue-Morse stream (Example
\ref{tm}), the paper folding stream (Example \ref{pap} below) 
and the Kolakoski stream (Example \ref{kol}) are predefined. 
For all of these examples termination of the
observational variant can be proved fully automatically both by
AProVE \cite{aprove} and TTT2 \cite{ttt2}, proving
well-definedness of the given stream specification.
\end{enumerate}
Apart from this graphical Windows version there is also a command line version to
be run under Linux. This provides the main facility, that is, generates the 
observational variant in term rewriting format in case the syntax 
is correct, and generates an appropriate error message otherwise.

None of the actions require substantial computation: for all features the result 
shows up instantaneously. On the other hand, proving termination of a resulting
observational variant by a tool like AProVE or TTT2 may take some computation time, 
although never more 
than a few seconds for the given examples. This was one of the objectives of the
project:
the transformation itself should be simple and direct, while the real work to be done
makes use of the power of current termination provers.

We conclude this section by an interesting stream specification that can be
dealt with by our implementation. Just like in the introduction for $\fib$,
and later in Section \ref{secex} we also show a turtle visualization. These
and others are made by a few lines of code traversing a boolean array containing the
first $N$ elements of a stream. These first $N$ elements are determined by
executing outermost rewriting with respect to $R_s$ starting in the 
constant representing the intended stream, until the first $N$ elements have
been computed.

\begin{example}
\label{pap}
Start by a ribbon of paper. Fold it half lengthwise. Next fold the folded ribbon
half lengthwise again, and repeat this a number of times, every time folding
in the same direction. Now by unfolding the ribbon one sees a sequence of
top-folds and valley-folds, and the question is what is the pattern in this
sequence. A first observation is that this pattern is the first half of the
pattern obtained when folding once more, so every such sequence is a proper
prefix of the next sequence. As a consequence, we can take the limit, obtaining 
a boolean stream $P$, called the {\em paper folding stream}, in which top 
folds and valley folds are represented by 0 and 1, respectively. 

Imagine what happens if we do an extra fold. Then all existing folds remain,
but between every two consecutive folds a new fold is created. These new
folds are alternately top folds and valley folds. So the effect of folding
once more is that the new sequence is the $\zip$ of $010101\cdots$ and the
old sequence. Taking the limit we obtain 
\[ P = \zip(\alt,P), \]
where for $\zip$ and $\alt$ we have the equations
\[ \zip(x:\sigma,\tau) = x : \zip(\tau,\sigma), \;\; 
\alt = 0:1:\alt. \] 

One may wonder whether $P$ is already fully defined by these three equations
for $P$, $\zip$ and $\alt$. It is, by Theorem \ref{main}, since the equations
form a proper stream
specification $R_s$ for which termination of $\Obs(R_s)$ is easily proved by 
TTT2 or AProVE.

Paper folding and many of its properties is folklore; we found this
characterization of $P$ independently.
Turtle visualization of $P$ is of particular interest, since the result is
not just a visualization, but also the shape obtained if the ribbon is not
fully unfolded, but only unfolded until the angles given as parameter of the 
turtle visualization. We only consider the case where the angles for 0 (top
fold) and 1 (valley fold) are equal. In case this angle is 90 degrees, then
the result is called the {\em dragon
curve}; this curve touches itself, but does not intersect itself. Pictures
are easily found on the Internet. When choosing turtle angles of less than 90
degrees, that is, the remaining paper fold is greater than 90 degrees, then
the curve neither touches nor intersects itself. Doing this for 87
degrees and doing 15 folds, this yields the following turtle visualization 
of the first $2^{15} - 1 = 32767$ elements of the stream $P$:

\includegraphics[height=3.6in,width=5.4in,angle=0]{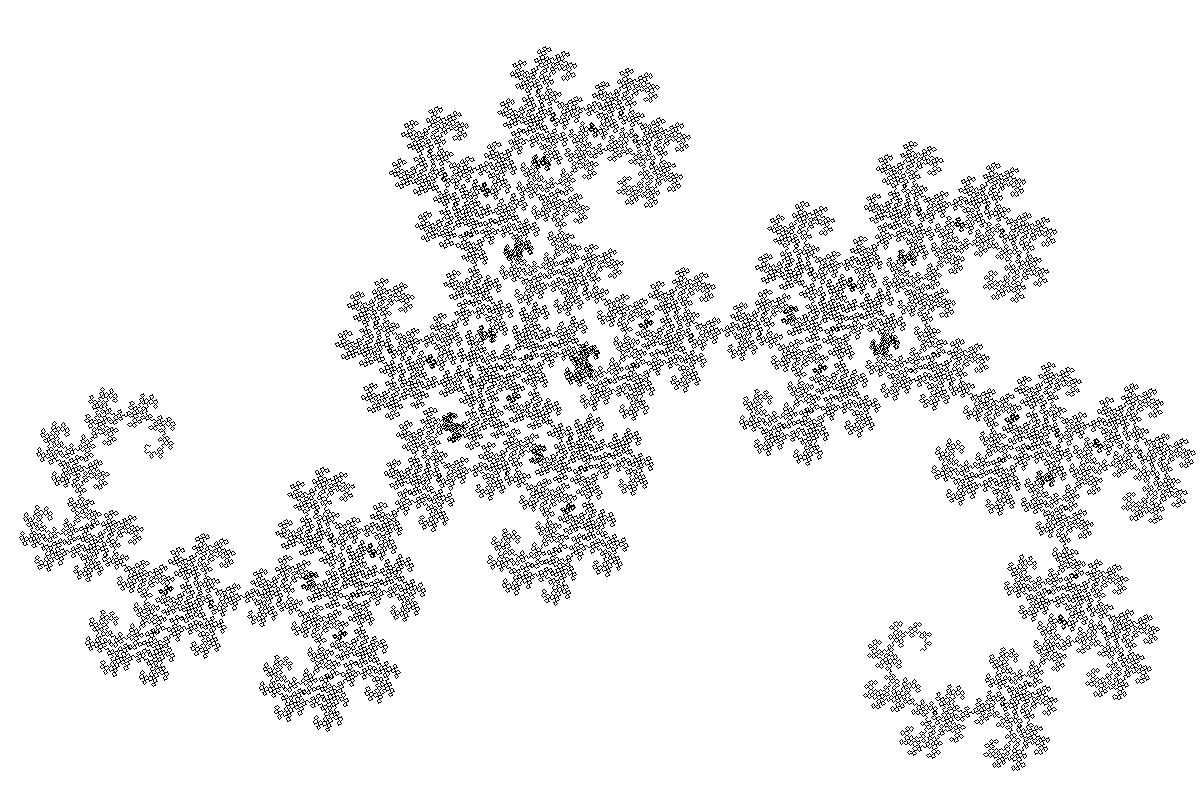}
\end{example}

\section{Data Independent Stream Functions}
\label{secdisf}

The reason that in Theorem \ref{main} we have to restrict to models satisfying
$S = \{[t] \mid t \in \TT \}$, as we saw in Example \ref{exgr}, is in the fact
that computations may be guarded by data elements in left-hand sides of equations. 
Next we show that we also get well-definedness for stream functions defined 
on all streams in case the left-hand sides of the equations do not contain data 
elements.

\begin{thm}
\label{main2}
Let $(\Sigma_d,\Sigma_s,R_d,R_s)$ be a proper stream specification for
which the TRS $\Obs(R_s) \cup R_d$ is terminating and the only subterms of 
left-hand sides of $R_s$ of sort $d$ are variables.  Then the stream
specification admits a unique model $(S,([f])_{f \in \Sigma_s})$
satisfying $S = \str$. 
\end{thm}

\begin{proof} (sketch)
We have to prove that for any $f \in \Sigma_s$ of type
$d^n \times s^m \to s$ the function $[f] : D^n \times (\str)^m \to \str$ is
uniquely defined. For doing so we introduce $m$ fresh constants $c_1,\ldots,c_m$
of sort $s$. Let $k \in \nat$ and $u_1,\ldots,u_n \in D$. Due to termination and 
orthogonality of  $\Obs(R_s) \cup R_d$, the term
$\hd(\tl^k(f(u_1,\ldots,u_n,c_1,\ldots,c_m)))$ has a unique normal form with
respect to $\Obs(R_s) \cup R_d$. Since it is of sort $d$, due to the shape of
the rules it is a ground term of sort $d$ over $\Sigma_d \cup \{\hd,
\tl,c_1,\ldots,c_m \}$, that is, a ground term $T$ composed from $\Sigma_d$ and
terms of the shape $\hd(\tl^i(c_j))$ for $i \in \nat$ and $j \in
\{1,\ldots,m\}$. For this observation it is essential that left-hand sides 
do not contain non-variable terms of sort $d$: terms of the shape
$f(\hd(\cdots),\cdots)$ should be rewritten. 

Let $N$ be the greatest number $i$ for which $T$ has a subterm
of the shape $\hd(\tl^i(c_j))$.  Let $s_1,\ldots,s_m \in \str$. Define
$t_j = s_j(0) : s_j(1) : \cdots : s_j(N) : \sigma$.
Since the term $\hd(\tl^k(f(u_1,\ldots,u_n,c_1,\ldots,c_m)))$ rewrites to $T$, 
$\hd(\tl^k(f(u_1,\ldots,u_n,t_1,\ldots,t_m)))$ \linebreak
rewrites to $T'$ obtained from
$T$ by replacing every subterm of the shape $\hd(\tl^i(c_j))$ by
$\hd(\tl^i(t_j))$. Observe that $\hd(\tl^i(t_j))$ rewrites to $s_j(i) \in D$. 
So \linebreak $([f](u_1,\ldots,u_n,s_1,\ldots,s_m))(k)$
has to be the $R_d$-normal form of
the ground term over $\Sigma_d$ obtained from $T$ by replacing every subterm of 
the shape $\hd(\tl^i(c_j))$ by $s_j(i) \in D$. Since this fixes 
$([f](u_1,\ldots,u_n,s_1,\ldots,s_m))(k)$ for every $k$, this uniquely defines 
$[f]$. 
%\qed
\end{proof}

\begin{example}
It is easy to see that for the standard stream functions $\zip$, $\even$ and $\odd$ 
defined by 
\[ 
\even(x:\sigma) = x : \odd(\sigma), \;\;
\odd(x:\sigma) = \even(\sigma), \;\;
\zip(x:\sigma,\tau) = x : \zip(\tau,\sigma), \] 
there exists $f : \str \to \str$ for every data set $D$ satisfying
\[ f(x : \sigma) \; = \; x:\zip(f(\even(\sigma)),f(\odd(\sigma))),\]
namely the identity. By Theorem \ref{main2} we can conclude it is the only one,
since for $R_d = \emptyset$ and $R_s$ consisting of the above four equations, the
resulting TRS $\Obs(R_s)$ consisting of the rules
\[ \begin{array}{rclrcl}
\hd(\even(\sigma)) & \to & \hd(\sigma)  &
\hd(\odd(\sigma)) & \to & \hd(\even(\tl(\sigma))) \\
\tl(\even(\sigma)) & \to & \odd(\tl(\sigma)) \hspace{1cm}&
\tl(\odd(\sigma)) & \to & \tl(\even(\tl(\sigma))) 
\end{array} \]
\[ \begin{array}{rcl}
\hd(f(\sigma)) & \to & \hd(\sigma) \\
\tl(f(\sigma)) & \to & \zip(f(\even(\tl(\sigma))),f(\odd(\tl(\sigma)))) 
\end{array} \]
and the rules for '$:$' and $\zip$ as in Example \ref{tmo},
is terminating as can be proved by AProVE \cite{aprove} or TTT2
\cite{ttt2}. Other approaches seem to fail: the technique from 
\cite{R05} fails to prove that the identity is the only stream function satisfying
the equation for $f$, while productivity of stream specifications containing
the rule for $f$ cannot be proved to be productive by the technique from
\cite{EGH08}. By essentially choosing $\Obs(R_s)$ as the input and adding 
information about special contexts, the tool Circ \cite{LGCR09} is able to 
prove that $f$ is the identity.
\end{example}

\section{More Transformations Preserving Semantics}
\label{secex}

Unfolding the Fibonacci stream specification as given in the introduction yields
the proper stream specification $R_s$ consisting of the equations
\[ \begin{array}{rclrcl}
\fib & = & f(\fib)&
g(0,\sigma) & = & 0:1:f(\sigma) \\
f(x:\sigma) & = & g(x,\sigma) \hspace{2cm} &
g(1,\sigma) & = & 0:f(\sigma). \end{array} \]
However, the TRS $\Obs(R_s)$ is not terminating since it allows the infinite reduction
\[ \tl(\fib) \to \tl(f(\fib)) \to \tl(g(\hd(\fib),\underbrace{\tl(\fib)})) \to
\cdots,\]
so our method fails to prove well-definedness of $\fib$ in a direct way.
In Lemma \ref{lemunf1} and Lemma \ref{lemunf2} we already saw two ways to modify
stream specifications while preserving their semantics. In this section we will
extend these lemmas to more general semantics preserving transformations, in
particular by making use of the equations $\EE$ from the proof of Lemma \ref{lem4}
that hold in every model. As an
example, we will apply such transformations to our original $\fib$ specification.
The observational variant of the resulting stream specification will be terminating,
so proving well-definedness of the transformed $\fib$ specification. But since 
the transformations are semantics preserving, this also proves well-definedness of
the original $\fib$ specification. 

In general we propose the following approach: in case for a stream specification the 
termination tools fail to prove termination of the observational variant, then try
to apply semantics preserving transformations as discussed in Section \ref{unf}
and this section until 
a transformed system has been found for which termination of the observational
variant can be proved. If this succeeds, this not only proves well-definedness of
the transformed specification, but also of the original one. 

In this approach we
have a symbol $\tl$ in several variants of the specification, while in the
construction of observational variant a fresh symbol $\tl$ is required. So
in the observational variant two versions of $\tl$ occur: the original symbol $\tl$ 
and the symbol $\tl$ created by $\Obs$. However, if the observational variant
happens to be terminating after identifying these two versions of $\tl$, then it
is also terminating if they are distinguished, so identifying them will not yield 
wrong results. But it may happen that
termination holds if the two versions of $\tl$ are distinguished, and does not hold
if they are identified. This is the case for Example \ref{tmo}.

Recall that mapping a stream specification $(\Sigma_d,\Sigma_s,R_d,R_s)$
to $(\Sigma_d,\Sigma'_s,R_d,R'_s)$ with $\Sigma_s \subseteq \Sigma'_s$ is said 
to {\em preserve semantics} if 
\begin{enumerate}[$\bullet$]
\item $(\Sigma_d,\Sigma_s,R_d,R_s)$ is well-defined if and only if 
$(\Sigma_d,\Sigma'_s,R_d,R'_s)$ is well-defined, and
\item If $(\Sigma_d,\Sigma_s,R_d,R_s)$ is well-defined with corresponding model 
$(S,[\cdot])$, and $(S',[\cdot]')$ is the model corresponding to 
$(\Sigma_d,\Sigma'_s,R_d,R'_s)$, then $[t] = [t]'$ for
all ground terms of sort $s$ over $\Sigma_d \cup \Sigma_s$.
\end{enumerate}
For well-definedness we required the model to satisfy
$S = \{[t] \mid t \in \TT \}$. For this section we need one more technical
requirement: $S$ should be closed under $\tl$. In order not to change our
definitions, throughout this section we assume the
following extra assumptions to achieve this requirement:
\begin{enumerate}[$\bullet$]
\item the symbol $\tl$ is in $\Sigma_s$, and
\item the corresponding equation $\tl(x:\sigma) = \sigma$ is in $R_s$.
\end{enumerate}
Lemma \ref{lemunf2} states that in keeping the same signature $\Sigma_s$, 
replacing $R_s$ by $R'_s$ preserves semantics as long as convertibility with respect 
to $R_s$ coincides with convertibility with respect to $R'_s$. But as we are
interested in preservation of semantics, this syntactical convertibility requirement
may be weakened to a more semantic version: if $R_s$ and $R'_s$ do not have the 
same convertibility relation, but allow the same models, the same can be concluded.
So now for a set $R_s$ of equations we will introduce a congruence $\sim_{R_s}$
being weaker than $=_{R_s}$, but still preserving semantics.

Recall the set $\EE$ of equations 
\[ \EE = \left\{ \begin{array}{rcl}
\hd(x:\sigma) & = & x \\
\tl(x:\sigma) & = & \sigma \\
\sigma & = & \hd(\sigma) : \tl(\sigma)
\end{array} \right. \]
For a set $R_s$ of equations of sort $s$ we define the relation $\sim_{R_s}$ on
terms over $\Sigma_s \cup \Sigma_d \cup \{\hd\}$ inductively by
\begin{enumerate}[$\bullet$]
\item if $\ell = r$ is in $R_s$ then $\ell \sim_{R_s} r$,
\item $\sim_{R_s}$ is reflexive, symmetric and transitive,
\item if $C$ is a context and $\rho$ is a substitution and $t \sim_{R_s} t'$,
then $C[t \rho] \sim_{R_s} C[t' \rho]$,
\item if $\ell = r$ is in $\EE$ then $\ell \sim_{R_s} r$,
\item if $t,t'$ are terms that may contain a fresh variable $x$ of type $d$, and
$t[x := u] \sim_{R_s} t'[x := u]$ for every $u \in D$, then $t \sim_{R_s} t'$.
\end{enumerate}
Note that $=_{R_s}$ is defined by the first three items, so $\sim_{R_s}$ generalizes
$=_{R_s}$ by the additional last two items.

\begin{lem}
\label{lemunf3}
Let $(\Sigma_d,\Sigma_s,R_d,R_s)$ and $(\Sigma_d,\Sigma_s,R_d,R'_s)$ be stream 
specifications satisfying $\ell \sim_{R'_s} r$ for all $\ell = r$ in $R_s$, and
$\ell \sim_{R_s} r$ for all $\ell = r$ in $R'_s$. 
Then transforming $(\Sigma_d,\Sigma_s,R_d,R_s)$ to
$(\Sigma_d,\Sigma_s,R_d,R'_s)$ preserves semantics.
\end{lem}
\begin{proof}
In an arbitrary model $(S,[\cdot])$ for a stream specification we define 
$[\hd](s) = s(0)$ for $s \in S \subseteq \str$. By assuming the equation
$\tl(x:\sigma) = \sigma$ we conclude $([\tl](s))(i) = s(i+1)$ for $i \geq 0$.
Combined with the definition of $[:]$ we conclude that $\EE$ holds in every model.

In case an equation $t[x := u] = t'[x := u]$ holds in a model for every $u \in D$, then by
definition the equation $t = t'$ holds in the model, too. 

Combining these observations we conclude by induction on the structure of $\sim_{R_s}$
that if a
model satisfies $R_s$, and $t  \sim_{R_s} t'$, then the model satisfies the equation
$t=t'$ too. Applying this both for $\sim_{R_s}$ and $\sim_{R'_s}$,
and using the conditions of the lemma we conclude that a model 
satisfies $(\Sigma_d,\Sigma_s,R_d,R_s)$ if and only if it 
satisfies $(\Sigma_d,\Sigma_s,R_d,R'_s)$. From this the lemma follows.
\end{proof}

\begin{example}
\label{fibex} 
Our $\fib$ specification completed by the $\tl$ equation reads
\[ \begin{array}{rclrcl}
f(0:\sigma) & = & 0:1:f(\sigma) &
\fib & = & f(\fib) \\ 
f(1:\sigma) & = & 0:f(\sigma) &
\hspace{1cm} \tl(x:\sigma) & = & \sigma. \end{array}\]
By Theorem \ref{thmunf} we know that unfolding this to 
\[ \begin{array}{rclrcl}
f(x:\sigma) & = & g(x,\sigma) &
\fib & = & f(\fib), \\ 
g(0,\sigma) & = & 0:1:f(\sigma) &
\hspace{1cm} \tl(x:\sigma) & = & \sigma  \\
g(1,\sigma) & = & 0:f(\sigma) & & & 
\end{array}\]
preserves semantics. Moreover, by Lemma \ref{lemunf1} we may add a constant $c$ to
the signature and add the equation $c = \tl(\fib)$, still preserving semantics. Let 
$R_s$ consist of these equations, and let $R'_s$ consist of
\[ \begin{array}{rclrcl}
f(x:\sigma) & = & g(x,\sigma) &
\fib & = & 0:c \\ 
g(0,\sigma) & = & 0:1:f(\sigma) &
c & = & 1:f(c) \\ 
g(1,\sigma) & = & 0:f(\sigma) &
\hspace{1cm} \tl(x:\sigma) & = & \sigma. \end{array}\]
Now we will check the conditions of Lemma \ref{lemunf3}.

For proving that $\ell \sim_{R'_s} r$ for all $\ell = r$ in $R_s$ we only need to 
consider the equations $c = \tl(\fib)$ and $\fib = f(\fib)$. We obtain
\[ c \; =_{R'_s} \; \tl(0:c) \;  =_{R'_s} \; \tl(\fib)\]
and
\[ \fib \; =_{R'_s} \; 0:c \; =_{R'_s} \; 0:1:f(c) \; =_{R'_s} \; g(0,c) \;  =_{R'_s} 
\; f(0:c) \; =_{R'_s} \; f(\fib).\]

For proving $\ell \sim_{R_s} r$ for all $\ell = r$ in $R'_s$ we only need to
consider the equations $\fib = 0:c$ and $c = 1:f(c)$. For this we need the congruence
$\sim_{R_s}$ rather than $=_{R_s}$. First observe
\[ \hd(g(0,\sigma)) \sim_{R_s} \hd(0:1:f(\sigma)) \sim_{R_s} 0, \]
and
\[ \hd(g(1,\sigma)) \sim_{R_s} \hd(0:f(\sigma)) \sim_{R_s} 0, \]
so by the last item of the definition of $\sim_{R_s}$ we obtain 
$\hd(g(x,\sigma)) \sim_{R_s} 0$.  Using this we get
\[ \begin{array}{rcl}
\fib & \sim_{R_s} & \hd(\fib) : \tl(\fib) \\
 & \sim_{R_s} & \hd(\fib) : c \\
 & \sim_{R_s} & \hd(f(\fib)) : c \\
 & \sim_{R_s} & \hd(f(\hd(\fib):\tl(\fib))) : c \\
 & \sim_{R_s} & \hd(g(\hd(\fib),\tl(\fib))) : c \\
 & \sim_{R_s} & 0 : c. \end{array} \]
Using $\fib \sim_{R_s} 0:c$, for the remaining equation we have
\[ \begin{array}{rcl}
c & \sim_{R_s} & \tl(\fib) \\
 & \sim_{R_s} & \tl(f(\fib)) \\
 & \sim_{R_s} & \tl(f(0:c)) \\
 & \sim_{R_s} & \tl(g(0,c)) \\
 & \sim_{R_s} & \tl(0:1:f(c)) \\
 & \sim_{R_s} & 1:f(c). \end{array} \]
So the requirements of Lemma \ref{lemunf3} are fulfilled and we conclude that
transforming $R_s$ to  $R'_s$ preserves semantics. By tools like 
AProVE or TTT2 one proves that $\Obs(R'_s)$ is terminating, so by Theorem
\ref{main} $R'_s$ is well-defined. Due to preservation of semantics the same holds
for $R_s$, and for the original $\fib$ specification.
\end{example}

As a consequence, we can conclude incompleteness of Theorem \ref{main}: the stream 
specification $R_s$ is well-defined but $\Obs(R_s)$ is not terminating, due to the
infinite reduction of $\Obs(R_s)$ we saw before.

The argument for the $\fib$ example can be given in a more sloppy way as was done
in \cite{Z09} as follows.
Identify ground terms with their interpretations in a model. The result of
$g$ always starts by $0$, so we can write $\fib = f(\fib) = g(\cdots) = 0:c$
for some stream $c$. Using this equality $\fib = 0:c$ we obtain
\[ 0:c = \fib = f(\fib) = f(0:c) = 0:1:f(c),\]
so $c = 1:f(c)$. So any model for the original specification also satisfies $R'_s$ 
which is obtained by replacing the equation $\fib = f(\fib)$ by the two equations 
$\fib = 0:c$ and $c = 1:f(c)$.  As $R'_s$ satisfies our format and $\Obs(R'_s)$ is 
terminating we conclude well-definedness of $\fib$.

For justifying the steps $f(\fib) = g(\cdots) = 0:c$ in this argument 
we need the last two items of the definition of $\sim_{R_s}$: 
\begin{enumerate}[$\bullet$]
\item for the step $f(\fib) = g(\cdots)$ we need $\EE$ to
rewrite $\fib$ to a term with ":" on top, and 
\item for the step $g(\cdots) = 0:c$ we
need the case analysis on the data element in "$\cdots$" as expressed by the last
item in the definition of $\sim_{R_s}$, 
\end{enumerate}
exactly as we did in our detailed proof.
Note that the sloppy argument only shows that the new equations in $R'_s$
are implied by original equations, and not the other way around. The following
example shows that it is essential also to prove the other direction.

\begin{example}
Consider the proper stream specification $R_s$ consisting of
\[ \begin{array}{rclrcl}
f(x:\sigma,y:\tau) & = & g(x,y) & \zeros & = & 0 : \zeros \\
g(0,0) & = & \ones & \ones & = &  1:\ones \\
g(0,1) & = & \zeros & c & = & f(c,c) \\
g(1,x) & = & \zeros & \hspace{1cm} \tl(x:\sigma) & = & \sigma \end{array} \]
If $[c]$ starts with 0, then $[f(c,c)] = [g(0,0)] = [\ones]$ starts with 1, and
if $[c]$ starts with 1, then $[f(c,c)] = [g(1,1)] = [\zeros]$ starts with 0, so $R_s$
does not admit a model and is not well-defined. However, the proper stream
specification $R'_s$ obtained from $R_s$ by replacing $c = f(c,c)$ by
$c = f(f(c,c),c)$ is well-defined, while this new equation satisfies
$c =_{R_s} f(f(c,c),c)$. Well-definedness of $R'_s$ can be proved by proving
termination of $\Obs(R''_s)$, where $R''_s$ is obtained from $R'_s$ by replacing the
equation for $c$ by $c = \zeros$. The transformation from $R'_s$ to $R''_s$ satisfies the
requirements of Lemma \ref{lemunf3}; for checking this one shows that
\[ f(f(0:\sigma,0:\sigma),0:\sigma) \sim_{R'_s} f(g(0,0),0:\sigma) \sim_{R'_s} 
f(1:\ones,0:\sigma)  \sim_{R'_s} g(1,0)  \sim_{R'_s} \zeros \]
and 
\[ f(f(1:\sigma,1:\sigma),1:\sigma) \sim_{R'_s} f(g(1,1),1:\sigma) \sim_{R'_s} 
f(0:\zeros,1:\sigma)  \sim_{R'_s} g(0,1)  \sim_{R'_s} \zeros \]
by which from the last item of the definition of $\sim_{R'_s}$ one concludes
\[f(f(x:\sigma,x:\sigma),x:\sigma) \sim_{R'_s} \zeros,\]
hence 
\[ c \sim_{R'_s} f(f(c,c),c) \sim_{R'_s}
  f(f(\hd(c) : \tl(c),\hd(c) : \tl(c)),\hd(c) : \tl(c)) \sim_{R'_s} 
\zeros.\]
\end{example}

Next we show how we can use the combination of Lemmas
\ref{lemunf1} and \ref{lemunf3} and Theorem \ref{main} to prove that the following 
stream specification admits exactly {\em two} models $(S,[\cdot])$ with 
$S = \{[t] \mid t \in \TT \}$:
\[ \begin{array}{rclrcl}
f(0:\sigma) & = & 0:1:f(\sigma) &
m & = & f(m) \\ 
f(1:\sigma) & = & 1:0:f(\sigma) &
\hspace{1cm} \tl(x:\sigma) & = & \sigma. \end{array}\]
Assume we have a model $(S,[\cdot])$ of this specification. Then either
$[m](0) = 0$ or  $[m](0) = 1$. In the former case the equation $m = 0:\tl(m)$
holds, in the latter case the equation $m = 1:\tl(m)$ holds. First assume we are
in the former case. Then we may add the equation $m = 0:\tl(m)$. For this extended
system we will prove well-definedness. Note that the specification is 
not orthogonal, but for applying Lemmas \ref{lemunf1} and \ref{lemunf3} and Theorem
\ref{thmunf} this is not required. After applying Lemma \ref{lemunf1} and Theorem
\ref{thmunf} we arrive at the (non-proper) specification $R_s$ consisting of 
\[ \begin{array}{rclrcl}
f(x:\sigma) & = & g(x,\sigma) & m & = & f(m) \\ 
g(0,\sigma) & = & 0:1:f(\sigma) & m & = & 0 : \tl(m) \\ 
g(1,\sigma) & = & 1:0:f(\sigma) & c & = & \tl(m) \\ 
& & & \hspace{1cm} \tl(x:\sigma) & = & \sigma. \end{array}\]
Now we transform this to the proper specification $R'_s$ consisting of 
\[ \begin{array}{rclrcl}
f(x:\sigma) & = & g(x,\sigma) & m & = & 0 : c \\ 
g(0,\sigma) & = & 0:1:f(\sigma) & c & = & 1 : f(c) \\ 
g(1,\sigma) & = & 1:0:f(\sigma) & \hspace{1cm} \tl(x:\sigma) & = & \sigma. \end{array}\]
One easily checks that $\ell =_{R'_s} r$ for all equations $\ell = r$ in $R_s$ and 
conversely, so by Lemma \ref{lemunf3} (or even Lemma \ref{lemunf2}) one concludes 
that this transformation is semantics preserving. Since $\Obs(R'_s)$ is easily
checked to be terminating, this shows that adding $m = 0:\tl(m)$ to the original
specification yields exactly one model with $S = \{[t] \mid t \in \TT \}$.
By symmetry the same holds for the other case, where the equation $m = 1:\tl(m)$
is added. Without a proof we mention that the two solutions for $m$ are exactly 
the Thue-Morse stream $\ms$ from Example \ref{tm} and its inverse.

We conclude this section by an elaboration of the Kolakoski stream.
\begin{example}
\label{kol}
The {\em Kolakoski stream} $\kol$ is the unique fix point of $g$ defined by
\[ \begin{array}{rcl}
g(0 : \sigma) & = & 1 : 1 : f (\sigma) \\
g(1 : \sigma) & = & 1 : f (\sigma) \\
f(0 : \sigma) & = & 0 : 0 : g(\sigma) \\
f(1 : \sigma) & = & 0 : g(\sigma).
\end{array} \]
So both for $f$ and $g$ its result on a stream is defined as follows.
If a 1 is read, then a single symbol is produced, and if a 0 read, then two
copies of a symbol are produced. This producing is done in such a way that
the produced elements are alternately 0's and 1's, for $f$ starting with 0
and for $g$ starting with 1. Due to this procedure in some presentations 
instead of 0 the number 2 is written. 

Of course we have to prove that $g$ has a unique fix point $\kol$. 
Similar to what we saw for $\fib$, the fix point equation $\kol = g(\kol)$ 
causes non-termination in the observational variant so we cannot apply our
approach directly. In order to prove well-definedness, we follow
the same lines as we did for $\fib$, with the difference that now we do not start
by unfolding, but postpone unfolding to the end. 
Start by the four equations for
$f$ and $g$, and the equations $\kol = g(\kol)$ and $\tl(x:\sigma) = \sigma$.
According to Lemma \ref{lemunf1} addition of the equation $K = \tl(\tl(\kol))$
is semantics preserving. So let $R_s$ consist of all of these equations for 
$f,g,\kol, \tl,K$. We will transform this to $R'_s$ consisting of the equations
for $f,g,\tl$, and the two equations 
\[ \kol = 1:0:K,\;\; K = 0:g(K).\]
Applying unfolding (Theorem \ref{thmunf}) to $R'_s$ yields a proper stream
specification for which TTT2 and AProVE succeed in proving termination of the
observational variant, so by Lemma \ref{lemunf3} it remains to show that
$\ell \sim_{R'_s} r$ for all $\ell = r$ in $R_s$, and
$\ell \sim_{R_s} r$ for all $\ell = r$ in $R'_s$. For doing so, first we show that 
$\hd(g(0:\sigma)) \sim_{R_s} 1$ and $\hd(g(1:\sigma)) \sim_{R_s} 1$, so
$\hd(g(x:\sigma)) \sim_{R_s} 1$, and hence $\hd(g(\sigma)) \sim_{R_s} 
\hd(g(\hd(\sigma):\tl(\sigma))) \sim_{R_s} 1$. Similarly we obtain 
$\hd(f(\sigma)) \sim_{R_s} 0$. Using this we derive 
\[\kol \sim_{R_s} \hd(\kol):\tl(\kol) \sim_{R_s}  \hd(g(\kol)):\tl(\kol) \sim_{R_s} 
 1:\tl(\kol),\]
and 
\[\hd(\tl(\kol)) \sim_{R_s} \hd(\tl(g(\kol))) \sim_{R_s} 
\hd(\tl(g(1:\tl(\kol)))) \sim_{R_s} \]
\[ \hd(\tl(1:f(\tl(\kol))))  \sim_{R_s} \hd(f(\tl(\kol))) \sim_{R_s} 0\]
from which $\kol \sim_{R_s} 1:0:K$ follows. Moreover, we obtain
\[ \begin{array}{rcl} K & \sim_{R_s} & \tl(\tl(\kol)) \\
 & \sim_{R_s} & \tl(\tl(g(\kol))) \\
 & \sim_{R_s} & \tl(\tl(g(1:0:K))) \\
 & \sim_{R_s} & \tl(\tl(1:f(0:K))) \\
 & \sim_{R_s} & \tl(f(0:K)) \\
 & \sim_{R_s} & \tl(0:0:g(K)) \\
 & \sim_{R_s} & 0:g(K). \end{array} \]
For the other direction we have
\[ g(\kol) \sim_{R'_s} g(1:0:K) \sim_{R'_s} 1:f(0:K) \sim_{R'_s} 1:0:0:g(K)
\sim_{R'_s} 1:0:K \sim_{R'_s} \kol \]
and $K \sim_{R'_s} \tl(\tl(1:0:K)) \sim_{R'_s} \tl(\tl(\kol))$,
concluding the proof.

Although this stream $\kol$ has a very simple and regular definition, the stream 
seems to behave remarkably irregular. In contrast to earlier streams
we saw, turtle visualizations of $\kol$ show up hardly any regular pattern: they 
seem to behave just like randomly generated boolean streams. For instance, by
choosing the angle to be 90 degrees both for 0 and 1,
taking the first 50000 elements of $\kol$ yields the following turtle visualization:
\includegraphics[height=2.5in,width=5.5in,angle=0]{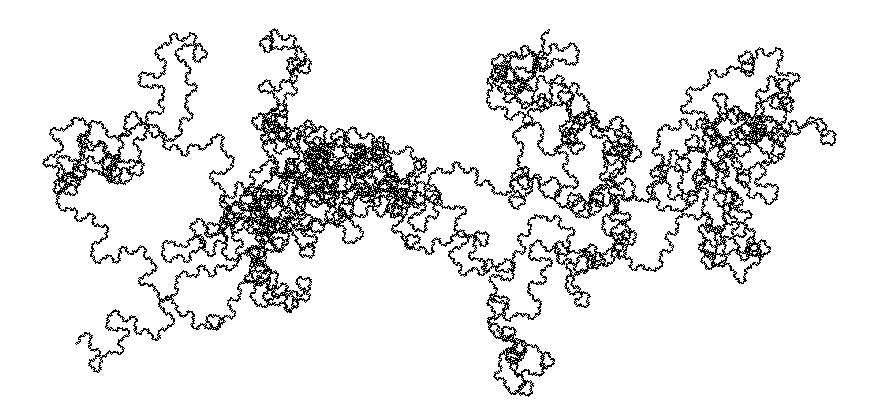}
\end{example}

\section{Conclusions and Further Research}
\label{secconcl}
We presented a technique by which well-definedness of stream specifications like
Example \ref{exgr} can be proved fully automatically, where a tool like Circ
\cite{LR07,LGCR09} fails, and
the productivity tool \cite{EGH08} fails to prove productivity. The main 
idea is to prove well-definedness by proving termination of a transformed system
$\Obs(R_s)$, in this way exploiting the power of present termination provers.

We observed that productivity of the stream specification cannot be concluded from
termination of $\Obs(R_s)$. Intuitively, productivity is closely related
to termination; we leave as a challenge to further relate termination with
productivity of stream specifications. A first step in this direction was
made in \cite{ZR09}. There it was proved that productivity of a stream
specification is equivalent to {\em balanced outermost termination} of
the specification extended by an extra rule $x : \sigma \to \ovf$. Here 
an infinite reduction is called balanced outermost if only outermost
redexes are reduced, and in the choice of them some fairness condition
holds. As
there are powerful techniques to prove outermost termination automatically 
\cite{EH09}, this can be used to prove productivity fully automatically. 
Unfortunately, as soon as binary operations like $\zip$ come in, typically
the notions outermost termination and balanced outermost
termination do not coincide: for many productive stream specifications the 
extension by the overflow rule is not outermost terminating, by which this
approach fails. Instead in \cite{ZR10} some basic criteria for productivity 
have been investigated together with relationship with context-sensitive
termination. Combined with a number of transformations and corresponding
heuristics, this yields a powerful technique for proving productivity 
automatically, supported by an implementation. This approach exploits the 
power of present termination provers for proving productivity, just like we do 
in this paper for proving well-definedness.

We offer an implementation for computing $\Obs(R_s)$ automatically, by which proving
well-definedness can be done fully automatically in case the approach applies 
directly. For cases for which the approach does not apply directly, 
in Section \ref{unf} and Section \ref{secex} we developed techniques to transform
stream specifications in such a way that semantics and well-definedness is
preserved, and often our approach applies to the transformed specifications.
Among these techniques only unfolding (Theorem \ref{thmunf}) is supported by
our implementation. For the other techniques some heuristics will be required. 
For the Fibonacci stream (Example \ref{fibex}) and the Kolakoski stream
(Example \ref{kol}) the following heuristics turned out to be successful:
\begin{enumerate}[$\bullet$]
\item Identify a non-productive constant $c$. In both mentioned examples this
is the stream to be defined, for which the equation is of the shape $c = f(c)$.
\item Determine the first element $d$ of the stream represented by $c$.
\item Introduce a fresh constant $c'$, and introduce the equation $c = d: c'$.
\item Using both the original equations and this new equation $c = d: c'$ try
to find a sound equation $c' = t$ in which $t$ is a term containing $c'$, but
not $c$.
\item Replace the original equation $c= \cdots$ by the two new equations 
$c = d: c'$ and $c' = t$, and check whether this transformation is semantics
preserving.
\item In case this approach fails, try the generalization in which the first
$n$ elements $d_1,\ldots,d_n$ of the stream represented by $c$ are determined
for some small value $n$, and the equation $c = d_1 : d_2 : \cdots : d_n : c'$
is introduced for a fresh constant $c'$.
\end{enumerate}

Another approach of using the techniques of Section \ref{unf} and Section \ref{secex}
is proving well-definedness of a stream specification by proving productivity
of all ground terms in a transformed specification, e.g., by the approach of 
\cite{ZR10}. Since 
productivity implies well-definedness and the transformation preserves semantics, this
implies well-definedness of the original specification. 

In Section \ref{secex} we used the technical assumption that the model is closed 
under $\tl$. This was forced by assuming the equation $\tl(x:\sigma) = \sigma$.
We conjecture that for the validity of the approach this is not essential. More precisely,
we conjecture that a stream specification $(\Sigma_d,\Sigma_s,R_d,R_s)$ with
$\tl \not\in \Sigma_s$ is well-defined if and only if the extended specification
$(\Sigma_d,\Sigma_s \cup \{\tl\},R_d,R_s \cup \{\tl(x:\sigma) = \sigma\})$ 
is well-defined. This looks trivial as $\tl$ does not occur in the original
specification, so is not expected to influence anything. However, giving a 
formal proof causes problems. The reason is that the 
model for $(\Sigma_d,\Sigma_s,R_d,R_s)$ may not be closed under $[\tl]$. In fact we can 
even prove that in the model $(S,[\cdot])$ for the $\fib$  example satisfying 
$S = \{[t] \mid t \in \TT \}$, the tail of $\fib$ is {\em not} contained in $S$.
A problem is how to lift $[f]$ defined on $S$ to the larger model that is
closed under $\tl$. For the particular $\fib$ example a solution can be
given, but for the general setting we failed.

This paper purely focuses on streams over a fixed data set $D$; in all examples 
even $D$ consists of the booleans. It is expected that the approach can be
generalized to other infinite data types like infinite binary trees. A
suitable format for this more general kind of infinite data structures has
been given in \cite{ZR10}. In such a setting destructors can be defined as
inverses of the constructors, just like in this paper we introduced
$(\hd,\tl)$ as the inverse of '$:$'. Similar to what we did in this paper for
streams, in this more general setting  a specification 
consisting of equations on terms over constructors and user defined symbols will
be transformed to an observational variant, being a rewrite system over destructors
and the user defined symbols. Just like we did in this paper for the special 
case of streams, this rewrite system serves for observing data. It is
orthogonal by construction, and well-definedness can be concluded from 
termination. Although the agenda for this approach for other infinite data 
structures is similar to what we did in this paper, this has not been
elaborated in detail.

\bibliography{ref}

\end{document}